\documentclass[11pt]{article}

\makeatletter
\def\maxwidth{ %
  \ifdim\Gin@nat@width>\linewidth
    \linewidth
  \else
    \Gin@nat@width
  \fi
}
\makeatother

\RequirePackage{amsmath,amssymb}
\RequirePackage{amsthm}
\RequirePackage{natbib}
\RequirePackage[colorlinks,allcolors=blue]{hyperref}

\usepackage{fullpage}
\usepackage{setspace}

\usepackage{bbm,graphicx,bm} %subfigure
\usepackage{enumitem}
\usepackage{multirow}
\usepackage[all,import]{xy}
\usepackage{appendix}
\usepackage{comment}
\usepackage{color}
\usepackage{soul}
\usepackage{pdflscape}
\usepackage{subcaption}

\usepackage{booktabs}
\usepackage{longtable}
\usepackage{array}
\usepackage[table]{xcolor}
\usepackage{wrapfig}
\usepackage{float}
\usepackage{colortbl}
\usepackage{pdflscape}
\usepackage{tabu}
\usepackage{threeparttable}
\usepackage{threeparttablex}
\usepackage[normalem]{ulem}
\usepackage{makecell}
\usepackage{afterpage}

\theoremstyle{definition}
\newtheorem{assumption}{Assumption}

\newtheorem{theorem}{Theorem}
\newtheorem{lemma}{Lemma}

\newtheorem{proposition}{Proposition}

\newcommand{\R}{\ensuremath{\mathbb{R}}}

\newcommand{\N}{\ensuremath{\mathbb{N}}}

\newcommand{\bbone}{\ensuremath{\mathbbm{1}}}

\newcommand{\E}{\ensuremath{\mathbb{E}}}
\newcommand{\calE}{\ensuremath{\mathcal{E}}}

\newcommand{\calR}{\ensuremath{\mathcal{R}}}

\newcommand{\calL}{\ensuremath{\mathcal{L}}}

\newcommand{\calP}{\ensuremath{\mathcal{P}}}

\newcommand{\Var}{\text{Var}}

\newcommand{\diag}{\text{diag}}

\newcommand{\mrp}{\text{mrp}}
\newcommand{\mr}{\text{mr}}
\newcommand{\drp}{\text{drp}}
\newcommand{\ps}{\text{ps}}
\def\super{\textsuperscript}

\def\b1{\boldsymbol{1}}

\newcommand\edit{}

\definecolor{RED}{RGB}{255,0,0}

\usepackage{soul}
\usepackage[utf8]{inputenc}

\title{Multilevel calibration weighting for survey data%
\thanks{We would like to thank Peng Ding, Yair Ghitza, Mark Handcock, Luke Keele, Shiro Kuriwaki, Drew Linzer, Luke Miratrix, Doug Rivers, Jonathan Robinson, Jas Sekhon, Betsy Sinclair, Brandon Stewart, and participants at BigSurv2020, AsianPolmeth VIII, and Polmeth 2021 for useful discussion and comments. This research was supported in part by the Hellman Family Fund at UC Berkeley and by the Institute of Education Sciences, U.S. Department of Education, through Grant R305D200010. The opinions expressed are those of the authors and do not represent views of the Institute or the U.S. Department of Education.}}

\author{Eli Ben-Michael, Avi Feller, and Erin Hartman}
\date{\today}

\begin{document}

\maketitle

\thispagestyle{empty}
\pagenumbering{gobble}

\begin{abstract}
\singlespacing
In the November 2016 U.S. presidential election, many state level public opinion polls, particularly in the Upper Midwest, incorrectly predicted the winning candidate. One leading explanation for this polling miss is that the precipitous decline in traditional polling response rates led to greater reliance on statistical methods to adjust for the corresponding bias---and that these methods failed to adjust for important interactions between key variables like education, race, and geographic region. 
Finding calibration weights that account for important interactions remains challenging with traditional survey methods: raking typically balances the margins alone, while post-stratification, which exactly balances all interactions, is only feasible for a small number of variables. 
In this paper, we propose multilevel calibration weighting, which 
enforces tight balance constraints for marginal balance and looser constraints for higher-order interactions.
This incorporates some of the benefits of post-stratification while retaining the guarantees of raking.
We then correct for the bias due to the relaxed constraints via a flexible outcome model; we call this approach Double Regression with Post-stratification (DRP).
We characterize the asymptotic properties of these estimators and show that the proposed calibration approach has a dual representation as a multilevel model for survey response. 
We then use these tools to to re-assess a large-scale survey of voter intention in the 2016 U.S. presidential election, finding meaningful gains from the proposed methods. The approach is available in the \texttt{multical} R package.
\end{abstract}

\clearpage
\pagenumbering{arabic}
\onehalfspacing
% \doublespacing

\section{Introduction}
Public opinion polling for the November 2016 U.S. presidential election was notable for the failure of state level polls, particularly in the Upper Midwest, to accurately predict the winning candidate. A leading explanation is that the precipitous decline in response rates for traditional polling approaches  led to increased reliance on possibly non-representative convenience samples. 
Such surveys can lead to large biases when analysts fail to adequately adjust for differences in response rates across groups, especially groups defined by fine-grained higher-order interactions
\citep{Kennedy2019, caughey_2020_elements}.
In particular, recent evaluations of 2016 election polling found that failing to adjust for the interaction between key variables such as education, race, and geographic region resulted in substantial bias \citep{Kennedy2018}.  

A pressing statistical question in modern public opinion research is therefore how to find survey weights that appropriately adjust for such higher-order interactions.  Traditional approaches, like raking, can perform poorly with even a moderate number of characteristics, typically balancing marginal distributions while failing to balance higher-order interactions. By contrast, post-stratification, which exactly balances all interactions, is only feasible for a small number of variables.  And while approaches like multilevel regression and post-stratification \citep[MRP;][]{Gelman1997} use outcome modeling to overcome this, they do not produce a single set of survey weights for all outcomes.
Fortunately, recent research on modern survey calibration \citep[e.g.,][]{Guggemos2010, chen2020dr_surveys} and on balancing weights for causal inference \citep[e.g.,][]{Zubizarreta2015, Hirshberg2019_augment} offer promising paths forward.

Building on these advances, we propose two principled approaches to account for higher-order interactions when estimating population quantities from non-probability samples.
First, we propose \emph{multilevel calibration} weighting, which exactly balances the first-order margins and \emph{approximately} balances interactions, prioritizing balance in lower-order interactions over higher-order interactions. 
Thus, this approach incorporates some of the benefits of post-stratification while retaining the guarantees of the common-in-practice raking approach.
And unlike outcome modeling approaches like MRP, multilevel calibration weights are estimated once and applied to all survey outcomes, an important practical constraint in many survey settings.

In some cases, however, multilevel calibration weighting alone may be insufficient to achieve good covariate balance on all higher order interactions, possibly leading to bias; or researchers might only be focused on a single outcome of interest. 
For this, we propose \emph{Double Regression with Post-stratification} (DRP), which combines multilevel calibration weights with outcome modeling. 
Similar to model-assisted survey calibration \citep{Breidt2017}, this approach uses outcome modeling, taking advantage of flexible modern prediction methods, to estimate and correct for possible bias from imperfect balance.
When the weights alone achieve good balance on higher-order interactions, the adjustment from the outcome model is minimal. When the higher-order imbalance is large, however, the bias correction will also be large and the combined estimator will rely more heavily on the outcome model. 

We characterize the numerical and statistical properties of both multilevel calibration weighting and the combined DRP estimator. 
By linking non-response bias to imbalance between the respondent sample and the overall population, we show how multilevel calibration and DRP control bias, and indicate a tradeoff between lower bias through better balance and lower variance through smaller adjustments.
We then describe this behavior asymptotically and show that the bias correction is critical for asymptotic Normality, yielding a non-parametric analog to recent double-robustness results in survey estimation \citep{chen2020dr_surveys}.
Through the Lagrangian dual, we also show that the multilevel calibration approach implicitly fits a multilevel model of non-response, shrinking the coefficients on higher order interaction terms and partially pooling across cells.

With these tools in hand, we return to the question of the failure of state-level polls in the 2016 US Presidential election. As we note above, \cite{Kennedy2018} show that many 2016 surveys failed to accurately account for the shift in public opinion among white voters with no college education, particularly in the Midwestern region of the country. 
We evaluate whether accounting for this higher-order interaction of race, education level, and geographic region can, retrospectively, improve public opinion estimates in the publicly available Pew poll. 
In particular, we combine a pre-election Pew poll of vote intention with the large, post-election Cooperative Congressional Election Study \citep[CCES;][]{DVN/GDF6Z0_2017}. We then assess how well the (re-wewighted) pre-election poll predicts the post-election ``ground truth.'' 
We also construct a simulation study calibrated to this example. 
We show that the multilevel weights substantially improve balance in interactions relative to raking and ad hoc post-stratification. 
We then show that further bias correction through DRP can meaningfully improve estimation.

Our proposed approach builds on two important advances in both modern survey methods and in causal inference. First, there has been a renewed push to find calibration weights that allow for \emph{approximate} balance on covariates, rather than exact balance \citep{Park2009, Guggemos2010, Zubizarreta2015}. Second, several recent approaches combine such weights with outcome modeling, extending classical generalized regression estimators in survey sampling and doubly robust estimation in causal inference \citep{chen2020dr_surveys, Athey2018_residual, Hirshberg2019_augment, Tan2020_calibrated}; we view our proposed DRP approach as a particular implementation of such augmented balancing weights. We give more detailed reviews in Sections \ref{sec:raking} and \ref{sec:multilevel_calibration}.

The paper proceeds as follows. Section \ref{sec:background} describes the notation and estimands, and formally describes various common survey weighting procedures such as raking and post-stratification. Section \ref{sec:multilevel_weights} characterizes the estimation error for arbitrary weighting estimators to motivate our multilevel calibration procedure, then describes the procedure. Section \ref{sec:drp} proposes the DRP estimator and analyzes its numerical and statistical properties. Section \ref{sec:dual} shows the dual relation between multilevel calibration and modelling non-response with a multilevel model. Section \ref{sec:sims} reports a simulation study calibrated to the case study, and Section \ref{sec:numerical} uses these procedures in the application. 
The methods we develop here are available in the \texttt{multical} R package.

\subsection{2016 U.S. Presidential Election Polling}
\label{sec:intro_application}
While national public opinion polls for the November 8, 2016 U.S. presidential election were, on average, some of the most accurate in recent public opinion polling, state-level polls were notable in their failure to accurately predict the winning candidate, particularly in the Upper Midwest. These state-level errors in turn led public opinion researchers to incorrectly predict the winner of the electoral college. 
\cite{Kennedy2018} attribute these errors to three main sources: (1) a late swing among undecided voters towards Trump, (2) failure to account for non-response related to education level, particularly among white voters, and (3) to a lesser degree, failure to properly predict the composition of the electorate.  

While all three of these concerns are important for survey practitioners, our analysis focuses on addressing concern (2) by allowing for deep interactions among important covariates, including race, education level, and region.  
To isolate this concern, we combine two high-quality surveys from before and after the 2016 election. We begin with the October 16, 2016 Pew survey of $2,062$ respondents, the final public release of Pew's election polling \citep{pew2016}.\footnote{Since our survey is from mid-October, we cannot account for concern (1) above, a late break towards Trump among undecided voters, which may contribute to remaining residual bias.} The primary outcome is respondents' ``intent to vote'' for each major party.
We combine this with the $44,909$ respondents in the 2016 Congressional Cooperative Election Study (CCES) post-election survey, a large, high-quality public opinion survey that accurately captures electoral outcomes at the state level \citep[see][]{DVN/GDF6Z0_2017}. Here, the primary outcome is respondents' ``retrospective'' vote for each major party. 

The combined Pew and CCES observations form a ``population'' of size $N=46,971$, where observations from the Pew survey are coded as respondents and observations from the CCES are coded as non-respondents.
Using this target, rather than the ground truth defined by the actual electoral outcomes, helps to address concern (3) above.
Specifically, the CCES validates voters against Secretaries of State voter files, allowing us to use known voters for whom we have measured auxiliary covariates to define our target population.

\begin{figure}[tb]
  \centering \includegraphics[width=.95\maxwidth]{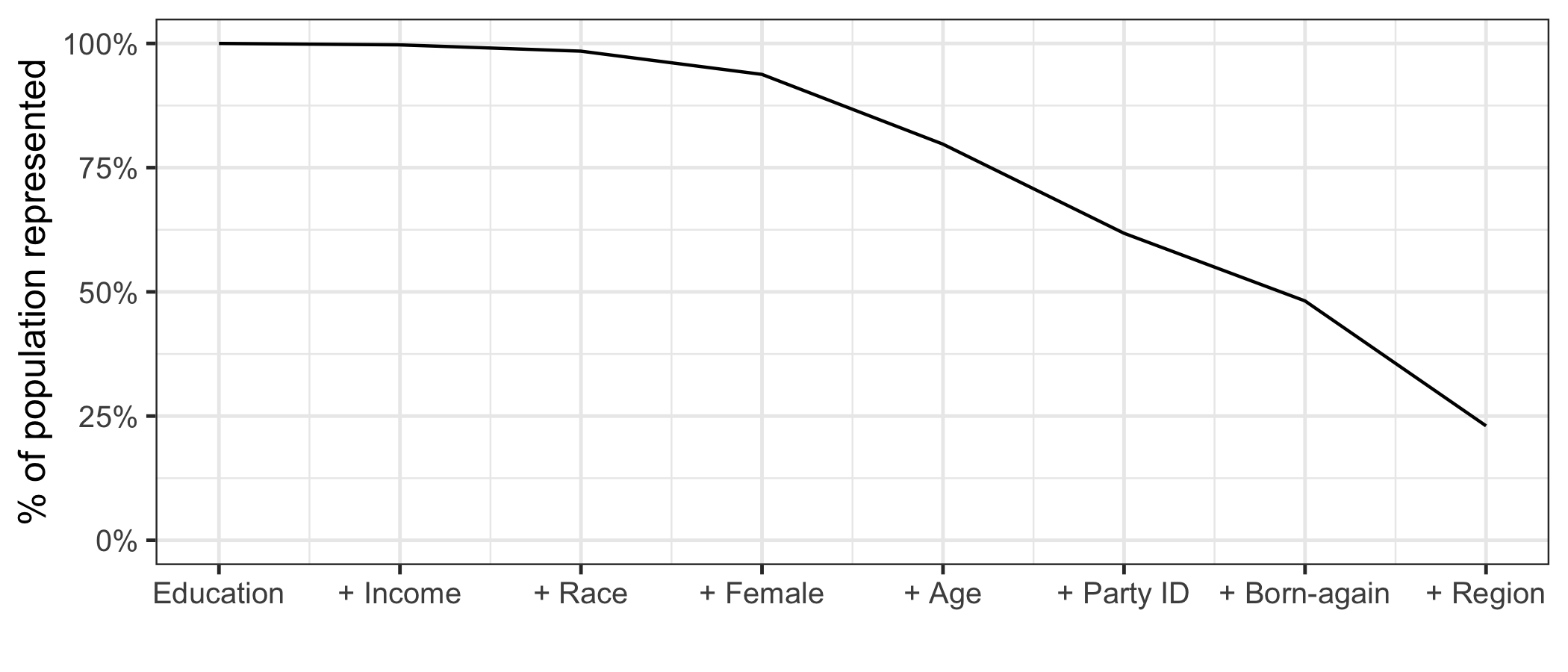}
\caption{Percent of the population that is not represented in the survey, beginning with education and successively interacting with income, religion, race, a binary for self-reported female, age, party identification, born-again Christian status, and region.}
\label{fig:n_empty_cell}
\end{figure}

Our goal is to adjust the respondent sample for possible non-response bias from higher-order interactions and assess whether the adjusted estimates are closer to the ground truth. 
As we discuss in Section \ref{sec:background} below, a key consideration is the set of voter characteristics that we adjust for via weighting and outcome modelling. Figure \ref{fig:n_empty_cell} shows the eight auxiliary variables we consider, measured in both the Pew and CCES surveys: income, religion, race, a binary variable for self-reported female, age, party identification, born-again Christian status, and region. All eight variables are coded as discrete, with the number of possible levels ranging from two to nine.\footnote{These are (i) education (6 levels), (ii) income (9 levels), (iii) race (4 levels), (iv) a binary for self-reported female (2 levels), (v) age (4 levels), (vi) party ID (3 levels), (vii) born again Christian (2 levels) and (viii) region (5 levels). Only 30\% of the potential combinations exist in the population.}  
Ideally, we would adjust for all possible interactions of these variables, via post-stratification. This is infeasible, however: there are $12,347$ possible combinations, far greater than the $n = 2,062$ respondents in our survey.
Figure \ref{fig:n_empty_cell} shows the percentage of the population that is represented in the survey as we progressively include---and fully interact---more covariates. With a single covariate, education (6 levels), each cell has at least one respondent. When including all eight covariates, the non-empty cells in the sample represent less than a quarter of the population.
These empty cells rule out using post stratification with these eight variables.
However, we do believe that there is important information in some higher order interactions, e.g. the three way interaction between race, education, and region.
This motivates our search for alternative adjustment methods that account for 
higher order interactions in a parsimonious way, prioritizing adjustment for strong interactions.

% \clearpage
\section{Background and setup}
\label{sec:background}
\subsection{Notation and estimands}
\label{sec:notation}
We will consider a finite population of $N$ individuals indexed $i=1,\ldots,N$. Each of these individuals has an outcome $Y_i$ that we observe if they respond to the survey, denoted by a binary variable $R_i$ where $R_i = 1$ indicates that unit $i$ responds, such as those units in our Pew sample; $n = \sum_i R_i$ is the total number of respondents. This response variable can include respondents to a probability sample or a convenience sample.\footnote{Often the response variable denotes whether a unit is in the sample or part of a target population, and so the respondents are not a subset of the overall population. For simplicity, we treat this as non-response, but our results can be extended to generalize to particular target populations.}
\edit{For a probability sample, $R_i = 1$ includes that unit $i$ was selected for the survey --- controlled by the analyst and part of the design --- and that unit $i$ responded --- outside of the analyst's control. For a convenience sample, $R_i=1$ simply denotes inclusion in the sample.}
In addition each individual is also associated with a set of $d$ categorical covariates $X_{i1},\ldots,X_{id}$, where the \edit{$\ell$\super{th}} covariate is categorical with \edit{$J_\ell$} levels, so that the vector of covariates $X_i \in [J_1] \times \ldots \times [J_d]$.

Rather than consider these variables individually, we will rewrite the vector $X_i$ as a single categorical covariate, the \emph{cell} for unit $i$, $S_i \in [J]$, where $J = J_1\times\ldots\times J_d$.\footnote{\edit{Note that there are always $J$ unique levels, but some may never appear in the target population}.} 
While we primarily consider a fixed population size $N$ and total number of distinct cells $J$, in Section \ref{sec:asymptotics} we will extend this setup to an asymptotic framework where both the population size and the number of cells can grow.
With these cells, we can summarize the covariate information. We denote $N^\calP \in \N^J$ as the \emph{population count vector} with $N^\calP_s = \sum_i \bbone\{S_i = s\}$, and $n^\calR \in \N^J$ as the \emph{respondent count vector}, with $n_s^\calR = \sum_i R_i \bbone\{S_i = s\}$. 
We will assume that we have access to these cell counts for both the respondent sample and the population.

Finally, for each cell $s$ we will consider a set of binary vectors $D_s^{(k)}$ that denote the cell in terms of its $k$\super{th} order interactions, and collect the vectors into matrices $D^{(k)} = [D_1^{(k)} \ldots D_J^{(k)}]'$, and into one combined $J \times J$ matrix $D = [D^{(1)},\ldots, D^{(d)}]$.\footnote{We focus on categorical covariates because it is common to only have population information for categorical covariates and so continuous covariates are often coarsened.
\edit{However, the procedures we describe below can be adapted for continuous covariates by incorporating more structure. For example, we can consider a polynomial basis expansion to include higher order moments, both marginally and jointly for interactions.}}
Figure \ref{fig:D_example} shows an example of $D^{(1)}$ and $D^{(2)}$ with three covariates: a binary for self-reported female, discretized age, and party identification \edit{(2, 4, and 3 levels, leading to $J = 24$ distinct cells)}. 
\edit{The (1 + 3 + 1 + 2) 7 columns of $D^{(1)}$ represent the margins of the three covariates, while the ($3 \times (1 + 2) + 1 \times 2$) 11 columns of $D^{(2)}$ represent the 2\super{nd} order interaction terms. Each row corresponds to a distinct  interaction between these three covariates, where the black areas represent elements of this matrix that are equal to 1. There are 6 remaining columns in the overall $24 \times 24$ matrix $D$, corresponding to $D^{(3)}$, the 3\super{rd} order interactions, not shown here due to space constraints.}

\begin{figure}[tb]
  \centering \includegraphics[width=.95\maxwidth]{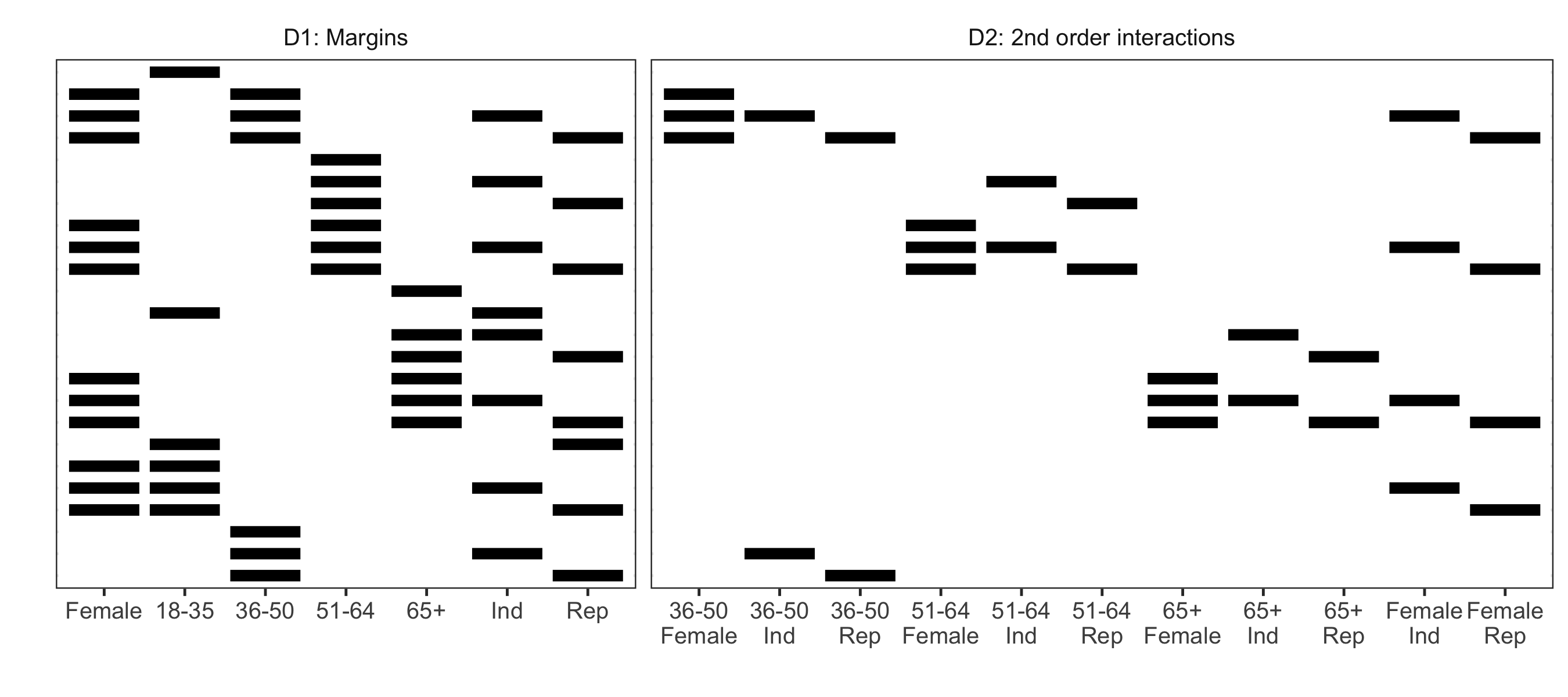}
\caption{Example of $D^{(1)}$ and $D^{(2)}$ with 3 covariates: age, a binary for self-reported female, and party identification.}
\label{fig:D_example}
\end{figure}

Our goal is to estimate the \emph{population} average outcome, which we can write as a cell-size weighted average of the within-cell averages, i.e.
\begin{equation}
  \label{eq:mu}
  \mu \equiv \frac{1}{N}\sum_{i=1}^N Y_i = \sum_{s=1}^J \frac{N^\calP_s}{N} \mu_s \;\;\; \text{ where } \;\;\; \mu_s \equiv \frac{1}{N_s^\calP}\sum_{S_{i} = s} Y_i.
\end{equation}
To estimate the population average, we will rely on the average outcomes we observe within each cell. For cell $s$, the responder average is 
\begin{equation}
  \label{eq:cell_avg}
  \bar{Y}_s \equiv \frac{1}{n_s^\calR}\sum_{S_i = s} R_i Y_i.
\end{equation}
We will assume that outcomes are missing at random within cells, so that the cell responder averages are unbiased for the true cell averages \citep{Rubin1976}:
\begin{assumption}[Missing at random within cells]
  \label{a:mar}
  For all cells $s=1,\ldots,J$, $\E\left[\bar{Y}_s \mid n_s^\calR \right] = \mu_s$.
\end{assumption}
\noindent We will denote the \emph{propensity score}  as $P(R_i = 1) \equiv \pi_i$, and the probability of responding conditional on being in cell $s$ as $\pi(s) \equiv \frac{1}{N_s^\calP}\sum_{S_{i} = s} \pi_i$. \edit{For a probability sample, $\pi_i$ denotes the joint probability of both selection into the survey and responding. The analyst knows and controls the selection probabilities but does not know the probability of response given selection. For a convenience sample  $\pi_i$ is the unknown probability of inclusion in the sample. For both cases the overall propensity score $\pi_i$ is unknown.} We assume that this probability is non-zero for all cells.
\begin{assumption}
  \label{a:overlap}
  % The \emph{propensity score} $P(R_i = 1) \equiv \pi_i > 0$ for all $i=1,\ldots,N$.
  $\pi(s) >0$ for all $s=1,\ldots,J$.
\end{assumption}
 \noindent These assumptions allow us to identify the overall population average using only the observed data. However, in order for Assumption \ref{a:mar} to be plausible, we will need the cells to be very fine-grained and have covariates that are quite predictive of non-response and the outcome. As we will see below, this creates a trade-off between identification and estimation: the former becomes more plausible with more fine-grained information, while the latter becomes more difficult \citep[see also][]{damour2020overlap}.

\subsection{Review: Raking and post-stratification}
\label{sec:raking}

We first consider estimating $\mu$ by taking a weighted average of the respondents' outcomes, with weights $\hat{\gamma}_i$ for unit $i$. Because the cells are the most fine-grained information we have, we will restrict the weights to be constant within each cell, relying on Assumption \ref{a:mar} that outcomes are missing at random within cells. We denote the estimated weight for cell $s$ as $\hat{\gamma}(s)$, and estimate the population average $\mu$ via:
\begin{equation}
  \label{eq:weight_est}
  \hat{\mu}\left(\hat{\gamma}\right) \equiv \frac{1}{N}\sum_{i = 1}^N R_i \hat{\gamma}_i Y_i = \frac{1}{N}\sum_s n_s^\calR \hat{\gamma}(s) \bar{Y}_s.
\end{equation}

\noindent If the individual probabilities of responding were known, we could choose to weight cell $S$ by the inverse of the propensity score, $\gamma(s) = \frac{1}{\pi(s)}$ \citep{Horvitz1952}. Unfortunately, the response probabilities are unknown. Instead, one class of procedures estimates the propensity score as $\hat{\pi}(s)$ and weights cell $s$ by the inverse \emph{estimated} propensity score $\hat{\gamma}(s) = \frac{1}{\hat{\pi}(s)}$.

One way to estimate the propensity score is as the proportion of the population in cell $s$ that responded, $\frac{n_s^\calR}{N_s^\calP}$, this leads to \emph{post-stratification} weights $\hat{\gamma}^\ps(s) = \frac{N_s^\calP}{n_s^\calR}$. Because the outcomes are missing at random within each cell by Assumption \ref{a:mar}, these post-stratification weights will lead to an unbiased estimator for the population average $\mu$. However, this estimator is only defined if there is at least one responder within each cell, and therefore it is often infeasible in practice. In even moderate dimensional cases it is unlikely there is at least one respondent for each cell. As we show for our application in Figure \ref{fig:n_empty_cell}, after including all eight covariates, less than a quarter of the population is represented by the non-empty cells in the sample.

One common alternative chooses weights so that the weighted \emph{marginal} distribution of the covariates exactly matches that of the full population. This \emph{raking on margins} procedure solves a convex optimization problem that finds the minimally ``disperse'' weights that satisfy this balance constraint \citep{Deming1940,Deville1992,Deville1993}.
Specifically, we find weights that solve
\begin{equation}
  \label{eq:raking}
  \begin{aligned}
    \min_{\gamma} \;\; & \sum_{s} n_s^\calR \gamma(s)^2\\
    \text{subject to } & \sum_s D_s^{(1)} n_s^\calR \gamma(s) = \sum_s D_s^{(1)} N_s^\calP\\
    & L \leq \gamma(s) \leq U \;\;\; \forall \;\; s = 1,\ldots,J.
  \end{aligned}
\end{equation}
These are the minimum variance weights that exactly balance first order margins. In addition to the variance penalty in the objective, we constrain the weights to be between a lower bound $L$ and an upper bound $U$. Setting the lower bound $L = 0 $ and the upper bound $U = \infty$ restricts the normalized cell weights, $\frac{1}{N}\gamma(s) n_s^\calR$ to be in the $J-1$ simplex. This ensures that the imputed cell averages are in the convex hull of the respondents' values and so do not extrapolate, and that the resulting estimator $\hat{\mu}(\hat{\gamma})$ is between the minimum and maximum outcomes in the sample.\footnote{
If we allow unbounded extrapolation and set $L = -\infty$ and $U = \infty$, the resulting estimator will be equivalent to linear regression weights from a linear regression of the outcome on first order indicators $D^{(1)}$ \citep{benmichael2020_ascm}. 
Many other choices of constraints and penalty functions are possible, including penalizing the distance to known design weights, see \citet{Deville1992}; these can also be incorporated into the discussion below.}
In contrast to the post-stratification weights, we typically expect weights to exist that exactly balance on the marginal distributions. However, failing to adjust for higher order interaction terms has the potential to induce severe bias.
\edit{Several papers have proposed ``soft'' or ``penalized'' calibration approaches to relax the exact calibration constraint in Equation \eqref{eq:raking}, allowing for approximate balance in some covariates \citep[see, e.g.][]{Huang1978,Rao1997,Park2009,Guggemos2010}. Our multilevel calibration approach below can be seen as adapting the soft calibration approach to full post-stratification.}

Before turning to our proposal for balancing higher-order interactions, we briefly describe some additional approaches.  \citet{chen2019calibrating} and \citet{mcconville2017model} discuss model-assisted calibration approaches which rely on the LASSO for variable selection.  \citet{caughey2017target} select higher-order interactions to balance using the LASSO.  \citet{hartman2021kpop} provide a kernel balancing method for matching joint covariate distributions between non-probability samples and a target population. \citet{linzer2011reliable} provides a latent class model for estimating cell probabilities and marginal effects in highly-stratified data.

Finally, in Section \ref{sec:drp}, we also discuss outcome modeling strategies as well as approaches that combine calibration weights and outcome modeling.
A small number of papers have previously explored this combination for non-probability samples. Closest to our setup is \citet{chen2020dr_surveys}, who combine inverse propensity score weights with a linear outcome model and show that the resulting estimator is doubly robust. Related examples include \citet{yang2020doubly}, who give high-dimensional results for a related setting; and \citet{si2020bayesian}, who combine weighting and outcome modeling in a Bayesian framework.

%%%
%%% MULTILEVEL CALIBRATION WEIGHTS
%%%
\section{Multilevel calibration: approximate post-stratification}
\label{sec:multilevel_weights}

We now propose \emph{multilevel calibration} weights, which bridge the gap between post-stratification and raking on the margins. First, we inspect the finite-sample estimation error and mean square error of the weighting estimator $\hat{\mu}(\hat{\gamma})$ for a set of weights $\hat{\gamma}$ \edit{that are deterministic functions of the cell counts $n^\calR$}, differentiating the impact of imbalance in lower- and higher-order terms on the bias. We then use this decomposition to find weights that control the components of the MSE by \emph{approximately} post-stratifying while maintaining raking on the margins.

\subsection{Estimation error}
\label{sec:estimation_error}

We begin by inspecting the estimation error $\hat{\mu}(\gamma) - \mu$ for weights $\gamma$. Define the \emph{residual} for unit $i$ as $\varepsilon_i \equiv Y_i - \mu_{S_i}$, and the average respondent residual in cell $s$ as $\bar{\varepsilon}_s = \frac{1}{n_s^\calR}\sum_{S_i = s}R_i \varepsilon_i$. The estimation error decomposes into a term due to imbalance in the cell distributions and a term due to idiosyncratic variation within cells: 
\begin{equation}
  \label{eq:est_error}
  \hat{\mu}\left(\hat{\gamma}\right) - \mu = \underbrace{\frac{1}{N}\sum_{s}\left(n_s^{\calR}\hat{\gamma}(s) - N_s^\calP\right) \times \mu_s}_{\text{imbalance in cell distribution}} + \underbrace{\frac{1}{N}\sum_{s}n_s^\calR \hat{\gamma}(s) \bar{\varepsilon}_s}_{\text{idiosyncratic error}}.
\end{equation}
By Assumption \ref{a:mar}, which states that outcomes are missing at random within cells, the idiosyncratic error will be zero on average, and so the bias will be due to imbalance in the cell distribution. By H\"{o}lder's inequality, we can see that the mean square error, conditioned on the number of respondents in each cell, is 
\begin{equation}
  \label{eq:condl_mse}
  \begin{aligned}
    \E\left[\left(\hat{\mu}\left(\hat{\gamma}\right) - \mu\right) ^ 2 \mid n^\calR \right] & = \underbrace{\frac{1}{N^2}\left(\sum_s \left(n_s^\calR \hat{\gamma}(s) - N_s^\calP\right)\mu_s\right)^2}_{\text{bias}^2} + \underbrace{\sum_{s} \left(\frac{n_s^{\calR}}{N}\right)^2\hat{\gamma}(s)^2\sigma_s^2}_{\text{variance}}\\
    & \leq \frac{1}{N^2}\sum_s\mu_s^2  \times \underbrace{\sum_s\left(n_s^\calR \hat{\gamma}(s) - N_s^\calP\right)^2}_{\text{imbalance in cell distribution}} + \underbrace{\sigma^2 \sum_{s} \left(\frac{n_s^{\calR}}{N}\right)^2\hat{\gamma}(s)^2}_{\text{noise}},
  \end{aligned}
\end{equation} 
where $\sigma^2_s = \Var(\bar{Y}_s \mid n^\calR)$ and $\sigma^2 = \max_s \sigma^2_s$. 
We therefore have two competing objectives if we want to control the mean square error for any given realization of our survey.
To minimize the bias we want to find weights that control the imbalance between the true and weighted proportions within each cell. 
To minimize the variance we want to find \edit{``diffuse''} weights so that the sum of the squared weights is small.

The decomposition above holds for imbalance measures across all of the strata, without taking into account their multilevel structure. 
In practice, we expect cells that share features to have similar outcomes on average.
We can therefore have  finer-grained control by 
leveraging our representation of the cells into their first order marginals $D_s^{(1)}$ and interactions of order $k$, $D_s^{(k)}$.
To do this, consider the infeasible population regression using the $D_s^{(k)}$ representation as regressors,
\begin{equation}
  \label{eq:pop_reg}
  \min_\eta \sum_{i=1}^N \left(Y_i - \sum_{k=1}^d D_{S_i}^{(k)} \cdot \eta_k \right)^2.
\end{equation}
With the solution to this regression, $\eta^\ast = (\eta_1^\ast,\ldots,\eta_d^\ast)$, we can decompose the population average in cell $s$ based on the interactions between the covariates, $\mu_s = \sum_{k=1}^d D_s^{(k)} \cdot \eta_k^\ast$. 
This decomposition in terms of the multilevel structure allows us to understand the role of imbalance in lower- and higher-order interactions on the bias.
Plugging this decomposition into Equation \eqref{eq:condl_mse} we see that the bias term in the conditional MSE further decomposes into the level of imbalance for the $k$\super{th} order interactions weighted by the strength of the interactions:
\begin{equation}
  \label{eq:condl_mse_popreg}
  \begin{aligned}
    \E\left[\left(\hat{\mu}\left(\hat{\gamma}\right) - \mu\right) ^ 2 \mid n^\calR \right] & = \frac{1}{N^2}\left(\sum_{k=1}^d \eta^{\ast}_k \cdot \sum_s \left(n_s^\calR \hat{\gamma}(s) - N_s^\calP \right)D_s^{(k)}\right)^2 + \sum_{s} \left(\frac{n_s^{\calR}}{N}\right)^2\hat{\gamma}(s)^2\sigma_s^2\\
    & \leq \frac{1}{N^2}\left(\sum_{k=1}^d \left\|\eta^{\ast}_k\right\|_2 \left\|\sum_s \left(n_s^\calR \hat{\gamma}(s) - N_s\right)D_s^{(k)}\right\|_2\right)^2 + \sigma^2\sum_{s} \left(\frac{n_s^{\calR}}{N}\right)^2 \hat{\gamma}(s)^2.
  \end{aligned}
\end{equation}

Typically, we expect that the ``main effects'' will be stronger than any of the interaction terms and so the coefficients on the first order terms, $\|\eta^\ast_1\|_2$, will be large relative to the coefficients for higher order terms. This is why raking on the margins as in Equation \eqref{eq:raking}---which only controls the main effects and implicitly assumes an additive functional form---is often seen as a good approximation \citep{mercer2018weighting}. However, ignoring higher order interactions entirely can lead to bias. We therefore propose to find weights that prioritize main effects while still minimizing imbalance in interaction terms when feasible.

\subsection{Multilevel calibration} 
\label{sec:multilevel_calibration}
We now design a convex optimization problem that controls the components of the conditional MSE on the right hand side of Equation \eqref{eq:condl_mse_popreg}. 
\edit{To do this, we apply the ideas and approaches developed for approximate balancing weights \citep[e.g.][]{Zubizarreta2015,Athey2018_residual,Wong2018, Hirshberg2019,Wang2019,Tan2020_calibrated,Ning2020_cbps} to the problem of controlling for higher order interactions, using our MSE decomposition as a guide.}
%To do this, 
We find weights that control the imbalance in all interactions in order to control the bias, while penalizing the sum of the squared weights to control the variance. Specifically, we solve the following optimization problem:
\begin{equation}
  \label{eq:primal}
  \begin{aligned}
    \min_{\gamma \in \R^J} \;\; & \sum_{k=2}^d \frac{1}{\lambda_k} \left\|\sum_s D_s^{(k)} n_s^\calR \gamma(s) - D_s^{(k)} N_s^\calP \right\|_2^2 + \sum_{s} n_s^\calR \gamma(s)^2\\[0.5em]
  \text{subject to} \;\; & \sum_s D_s^{(1)} n_s^\calR \gamma(s) = \sum_s D_s^{(1)} N_s^\calP\\[0.5em]
  & L \leq \gamma(s) \leq U\;\;\; \forall s=1,\ldots J.
  \end{aligned}
\end{equation}

\noindent where the $\lambda_k$ are hyper-parameters, discussed below.

We can view this optimization problem as adding an additional objective to the usual raking estimator in Equation \eqref{eq:raking}, optimizing for higher order balance. 
As with the raking estimator, the multilevel calibration weights are constrained to \emph{exactly balance} first order margins. Subject to this exact marginal constraint, the weights then minimize the imbalance in $k$\super{th} order interactions for all $k=2,\ldots,d$.
In this way, the multilevel calibration weights approximately post-stratify by optimizing for balance in higher-order interactions rather than requiring exact balance in all interactions as the post-stratification weights do. 
Following the bias-variance decomposition in Equations \eqref{eq:condl_mse} and \eqref{eq:condl_mse_popreg}, this objective is penalized by the sum of the squared weights. In addition, we also constrain the weights to be between the user-defined lower and upper limits, $L$ and $U$. As with the raking estimator, we can also replace the sum of the squared weights with a different penalty function, including penalizing deviations from known sampling weights following \citet{Deville1992}.
Finally, we could solve a variant of Equation \eqref{eq:primal} without the multilevel structure encoded by the $D_s^{(k)}$ variables. This would treat cells as entirely distinct and perform no aggregation across cells while approximately post-stratifying. From our discussion in Section \ref{sec:estimation_error}, this would ignore the potential bias gains from directly leveraging the multilevel structure.

\subsubsection{Hyperparameter selection}
An important component of the optimization problem are the hyper-parameters $\lambda_k$ for $k=2,\ldots,d$. These hyper-parameters control the relative priority that balancing the higher-order interactions receives in the objective in an inverse relationship. If $\lambda_k$ is large, then the weights will be more regularized and the $k$\super{th} order interaction terms will be less prioritized. In the limit as all $\lambda_k \to \infty$, no weight is placed on any interaction terms, and Equation \eqref{eq:primal} reduces to raking on the margins. Conversely, if $\lambda_k$ is small, more importance will be placed on balancing $k$\super{th} order interactions. For example, if $\lambda_2 = 0$, then the optimization problem will rake on margins \emph{and} second order interactions. As all $\lambda_k \to 0$ we recover post-stratification weights, if they exist.\footnote{If we allow unbounded extrapolation and set $L = -\infty$ and $U = \infty$, the resulting estimator will be equivalent to the MRP estimate \eqref{eq:mrp}, with regularization hyper-parameters $\lambda^{(k)}$ \citep{benmichael2020_ascm}.}

The hyper-parameters define a bias-variance trade-off. 
Smaller values will decrease the bias by improving the balance on higher order interaction terms. This comes at the expense of increasing variance by decreasing the effective sample size. In practice, we suggest explicitly tracing out this trade-off, as we show in Figure \ref{fig:balance_v_neff}. For a sequence of potential hyper-parameter values $\lambda^{(1)},\lambda^{(2)},\ldots$, set all of the hyper-parameters to be $\lambda_k = \lambda^{(j)}$. We can then look at the two components of the objective in Equation \eqref{eq:primal}, plotting the level of imbalance $\sum_{k=2}^d\left\|\sum_s D_s^{(k)} n_s^\calR \gamma(s) - D_s^{(k)} N_s^\calP \right\|_2^2$ against the effective sample size  $n^{\text{eff}} = \frac{\left(\sum_s n_s^\calR \hat{\gamma}(s)\right)^2}{\sum_s n_s^\calR \hat{\gamma}(s)^2}$. After fully understanding this trade-off, practitioners can choose a common $\lambda$ somewhere along the curve. For example, in our analysis in Section \ref{sec:numerical}, we choose $\lambda$ to 
achieve 95\% of the potential balance improvement in higher-order terms of $\lambda = 0$ relative to raking.
Finally, note that the optimization problem is on the scale of the population counts \edit{rather than the population proportions}. This means that with a common hyper-parameter, the higher-order interactions --- which have lower population counts $D^{(k)'} N^\calP$ --- will have less weight by design.

%%%
%%% DRP
%%%
\section{Double regression with post-stratification (DRP)}
\label{sec:drp}
So far we have focused on multilevel calibration, with weights that exactly match the first order marginals between the sample and the full population, while approximately balancing higher order interactions. 
This approach is independent of the outcomes and so we can use a single set of weights to estimate the population average for multiple different outcomes.
However, in some cases it may not be possible to achieve good covariate balance on higher order interactions, meaning that our estimates may still be biased.
Similar to model-assisted survey calibration with design weights \citep{Breidt2017}, we can address this by specializing to a particular outcome and by explicitly using outcome information to estimate and correct for the bias.

We begin by reviewing outcome modeling, especially multilevel regression with post-stratification (MRP), and then propose double regression with post-stratification (DRP).

\subsection{Using an outcome model for bias correction}
\label{sec:outcome_modelling}

A common alternative to the weighting approaches above is to estimate the population average $\mu$ using estimates of the cell averages $\mu_s$ \citep{Gelman1997}. These approaches take modelled estimates of the cell averages, $\hat{\mu}_s$ and \emph{post-stratify} them to the population totals as
\begin{equation}
  \label{eq:mrp}
  \hat{\mu}^{\mrp} = \frac{1}{N}\sum_{s}N_s^{\calP} \hat{\mu}_s.
\end{equation}
By smoothing estimates across cells, outcome modeling gives estimates of $\hat{\mu}_s$ even for cells with no respondents, thus sidestepping
the primary feasibility problem of post-stratification.
We use the term \emph{Multilevel Regression with Post-stratification} (MRP) to refer to the broad class of outcome modelling approaches that post-stratify modelled cell averages, and discuss particular choices of outcome model in Section \ref{sec:mrp_to_drp} below.

With weights $\hat{\gamma}$, we can use MRP to estimate the bias due to imbalance in higher order interactions by taking the difference between the MRP estimate for the population and a hypothetical estimate with population cell counts $n_s^\calR \hat{\gamma}(s)$:
\begin{equation}
    \label{eq:bias_est}
    \widehat{\text{bias}} = \hat{\mu}^\mrp - \frac{1}{N}\sum_sn_s^\calR\hat{\gamma}(s) \hat{\mu}_s = \frac{1}{N} \sum_s \hat{\mu}_s \times \left(N_s^\calP - n_s^\calR \hat{\gamma}(s)\right).
\end{equation}
This uses the outcome model to collapse the imbalance in the $J$ cells into a single diagnostic.
Our main proposal is to use this diagnostic to correct for any remaining bias from the multilevel calibration weights.
We refer to the estimator as \emph{Double Regression with Post-Stratification} (DRP), as it incorporates two forms of ``regression''---a regression of the outcome $\hat{\mu}(s)$ and a  regression of response $\hat{\gamma}(s)$ through the dual problem, as we discuss in Section \ref{sec:dual} below. %\eqref{eq:dual}.
We construct the estimator using weights $\hat{\gamma}(s)$ and cell estimates $\hat{\mu}_s$ as
\begin{equation}
  \label{eq:drp}
  \begin{aligned}
    \hat{\mu}^\drp\left(\hat{\gamma}\right) & =  \hat{\mu}\left(\hat{\gamma}\right) & + &\;\frac{1}{N}\sum_s  \hat{\mu}_s  \times \underbrace{\left(N_s^\calP  - n_s^\calR \hat{\gamma}(s) \right)}_{\text{imbalance in cell } s}\\
    & =  \hat{\mu}^\text{mrp} & + & \;\frac{1}{N}\sum_s  n_s^\calR \hat{\gamma}(s)  \times  \underbrace{(\bar{Y}_s - \hat{\mu}_s)}_{\text{error in cell }s}.
  \end{aligned}
\end{equation}
The two lines in Equation \eqref{eq:drp} give two equivalent perspectives on how the DRP estimator adjusts for imbalance.
The first line begins with the multilevel calibration estimate $\hat{\mu}(\hat{\gamma})$ and then adjusts for the estimate of the bias using the outcome model $\frac{1}{N}\sum_s \hat{\mu}_s(N_S^\calP - n_s^\calR\hat{\gamma}(s))$.
If the population and re-weighted cell counts are substantially different in important cells, the adjustment from the DRP estimator will be large. On the other hand, if the population and re-weighted sample counts are close in all cells then $\hat{\mu}^\drp(\hat{\gamma})$ will be close to $\hat{\mu}(\hat{\gamma})$. In the limiting case of post-stratification where all the counts are equal, the two estimators will be equivalent, $\hat{\mu}^{\drp}(\hat{\gamma}^\ps) = \hat{\mu}(\hat{\gamma}^\ps)$. 
The second line instead starts with the MRP estimate, $\hat{\mu}^\text{mrp}$, and adjusts the estimate based on the error within each cell. If the outcome model has poor fit in cells that have large weight, then the adjustment will be large.

As we discuss next, this uses outcome information to improve the statistical properties of the estimator, at the expense of specializing it to a particular outcome.
More broadly, this estimator is a special case of augmented approximate balancing weights estimators \citep{Athey2018_residual,Hirshberg2019_augment,Tan2020_model_assisted} and is closely related to generalized regression estimators \citep{Cassel1976}, augmented IPW estimators \citep{Robins1994, chen2020dr_surveys}, and bias-corrected matching estimators \citep{Rubin1976}.

\subsection{Bias reduction and asymptotic normality}
\label{sec:asymptotics}
We now show that adjusting for imbalance with an outcome model reduces bias, and that this bias reduction allows for inference through asymptotic normality by ensuring that the bias is asymptotically smaller than the variance.
To see this, we can again inspect the estimation error. Analogous to Equation \eqref{eq:est_error}, the difference between the DRP estimator and the true population average is
\begin{equation}
  \label{eq:est_error_drp}
  \hat{\mu}^\drp\left(\hat{\gamma}\right) - \mu = \frac{1}{N}\sum_{s}\underbrace{\left(n_s^\calR\hat{\gamma}(s) - N_s^\calP\right)}_{\text{imbalance in cell }s } \times \underbrace{(\hat{\mu}_s - \mu_s)}_{\text{error in cell } s} + \underbrace{\frac{1}{N}\sum_{s=1}^sn_s^\calR \hat{\gamma}(s) \bar{\varepsilon}_s}_{\text{noise}}.
\end{equation}
Comparing to Equation \eqref{eq:est_error}, where the estimation error depends solely on the imbalance and the true cell averages, we see that the estimation error for DRP depends on the \emph{product} of the imbalance from the weights and the estimation error from the outcome model. Therefore, if the model is a reasonable predictor for the true cell averages, the estimation error will decrease.

To formalize this, we consider an asymptotic framework with a sequence of finite populations of size $N$, and let $N \to \infty$.
In this framework, we make several modifications to our setup. First, we strengthen Assumption \ref{a:overlap} to hold strictly for all population sizes $N$, so that $\min_s \pi(s) \geq \pi^\ast > 0$, and we allow the lower bound $\pi^\ast$ to change with the population size $N$.
This ensures that we have a strictly non-zero probability of having a respondent in each cell in all the populations we consider. 
We then allow the number of cells $J$ to grow with the population size $N$. Denoting $\kappa \equiv \|D^{-1}\|_2\|D\|_2$ as the condition number of the $J \times J$ matrix $D$, we restrict the number of cells so that $\frac{\kappa^2 J}{(\pi^\ast N)^\alpha}$ converges to a constant for a rate $0 \leq \alpha < 1$.
We also adjust the multilevel calibration procedure to approximately (rather than exactly) rake on margins
without regularization, ensuring that there is always a feasible solution for every finite population.
Finally, we restrict the response variables $R_i$ to be independent.
We detail these and other regularity assumptions on the design in Appendix \ref{sec:proofs}.

\begin{theorem}
\label{thm:asymp_normal}
If $\frac{\kappa^2 J}{(\pi^\ast N)^\alpha}$ converges to a constant for some $0 \leq \alpha < 1$, and $\sum_s \left(\hat{\mu}_s - \mu_s\right)^2 = o_p\left(\left(\pi^\ast N \right)^{-\alpha/2}\right)$, then under the regularity conditions in Assumption \ref{a:regularity}, the DRP estimator $\hat{\mu}^\drp(\hat{\gamma})$
is
\[
\hat{\mu}^\drp(\hat{\gamma}) - \mu = \frac{1}{N}\sum_{i=1}^N\frac{R_i}{\pi(S_i)}\varepsilon_i + o_p\left(\frac{1}{\pi^\ast \sqrt{N}}\right).
\]
Furthermore, if $\frac{1}{N \sqrt{V_N}} \sum_{i=1}^N \frac{R_i}{\pi(S_i)}\varepsilon_i \Rightarrow N(0,1)$ for $V_N = \Var\left(\frac{1}{N}\sum_{i=1}^N \frac{R_i}{\pi(S_i)}\varepsilon_i\right)$, then $$\frac{\hat{\mu}^\drp(\hat{\gamma}) - \mu}{\sqrt{V_N}} \Rightarrow N(0,1).$$
\end{theorem}

Theorem \ref{thm:asymp_normal} shows us that as long as the modelled cell averages estimate the true cell averages well enough,
the model and the calibration weights combine to ensure that 
the bias will be negligible relative to the variance, asymptotically. 
The rate at which the number of cells grows with the population size affects how well the modelled cell averages need to perform. If the number of cells is constant, the model needs only to be consistent. On the other hand, if the number of cells grows quickly then Theorem \ref{thm:asymp_normal} implicitly requires more structure on the outcomes, so that the model can estimate the cell averages well enough.
As we discuss in Section \ref{sec:estimation_error}, we expect main effects to be much stronger than higher order interaction terms in practice. Therefore, including a new covariate or using a finer discretization of continuous covariates will primarily impact the outcome model through these main effects, leading to a substantial amount of underlying structure even though the total number of cells is increasing. As we discuss in Section \ref{sec:discussion}, there are alternative ways to account for an increasing number of cells $J$ that may allow $J$ to grow more quickly relative to the population size $N$. We leave a thorough investigation of these alternatives to future work.

The asymptotic variance of the DRP estimator depends on the variance of the residuals $\varepsilon_i$, which we expect to have much lower variance than the raw outcomes. So asymptotically the DRP estimator will also have lower variance than the oracle Horvitz-Thompson estimator,
similar to other model-assisted estimators \citep{Breidt2017}.
However, note that the minimum cell response probability $\pi^\ast$ affects the quality of the asymptotic approximation; if individuals in some cells are very unlikely to respond, it will be difficult both to model those averages and achieve good balance from the respondents.
\edit{Theorem \ref{thm:asymp_normal} is analogous to recent double-robustness results in survey estimation, such as from \citet{chen2020dr_surveys}. Rather than estimating a parametric outcome and non-response model, 
we instead consider all interactions.}

Finally, in our analysis in Section \ref{sec:numerical}, we use Theorem \ref{thm:asymp_normal} to construct confidence intervals for the population total $\mu$. To do this, we start with a plug-in estimate for the variance,
\begin{equation}
  \label{eq:std_err}
  \hat{V} = \frac{1}{N^2} \sum_{i=1}^n R_i \hat{\gamma}(S_i)^2 (Y_i - \hat{\mu}_{S_i})^2.
\end{equation}
We then construct approximate level $\alpha$ confidence intervals via $\hat{\mu}^\drp(\hat{\gamma}) \pm z_{1-\alpha/2}\sqrt{\hat{V}}$, where $z_{1-\alpha/2}$ is the $1-\alpha/2$ quantile of a standard normal distribution.

\subsection{Choosing an outcome model in MRP and DRP}
\label{sec:mrp_to_drp}

The choice of outcome model is crucial for both MRP and DRP.
As we discussed in Section \ref{sec:multilevel_weights},
we often believe that the strata have important hierarchical structure where  main effects and lower-order interactions are more predictive than higher-order interactions.
We consider two broad classes of outcome model that accommodate this structure: multilevel outcome models, which explicitly regularize higher-order interactions;  
and tree-based models, which implicitly regularize higher-order interactions.

\paragraph{Multilevel outcome model.} 
We first consider multilevel models, which have a linear form as $\hat{\mu}_s^{\mr} = \hat{\eta}^\mr \cdot D_s$, where $\hat{\eta}^\mr$ are the estimated regression coefficients \citep{Gelman1997, Ghitza2013, gao2020improving}.
MRP directly post-stratifies these model estimates, using the coefficients to predict the value in the population:
\[
\hat{\mu}^\mrp = \hat{\eta}^\mr \cdot \frac{1}{N}\sum_s N_s^\calP D_s.
\]
In contrast, the DRP estimator only uses the coefficients to adjust for any remaining imbalance after weighting,
\[
\hat{\mu}^\drp(\hat{\gamma}) = \hat{\mu}(\hat{\gamma}) + \hat{\eta}^\mr \cdot \left(\frac{1}{N} \sum_s D_s \left(N_s^\calP - n_s^\calR \hat{\gamma}(s)\right)\right).
\]
This performs bias correction.
 When we use a multilevel outcome model, we can also view the DRP estimator as a weighting estimator $$\hat{\mu}^\drp(\hat{\gamma}) = \frac{1}{N}\sum_{s}\tilde{\gamma}(s) n_s^\calR \bar{Y}_s.$$
In particular, 
when using the \emph{maximum a posteriori} (MAP) estimate of a multilevel model in the corresponding DRP estimator, the outcome model directly adjusts the weights
\[
  \tilde{\gamma}(s) = \hat{\gamma}(s) + \left(N^\calP - \text{diag}(n^\calR) \hat{\gamma}\right)'D\left(D'\text{diag}(n^\calR)D + Q \right)^{-1} D_s,
\]
where $Q$ is the prior covariance matrix associated with the multilevel model \citep{Breidt2017}.
Importantly, while the multilevel calibration weights $\hat{\gamma}$ are constrained to be non-negative, the DRP weights allow for \emph{extrapolation} outside of the support of the data \citep{benmichael2020_ascm}.

\paragraph{Trees and general weighting methods.} 
More generally, we can consider an outcome model that smooths out the cell averages, using a weighting function between cells $s$ and $s'$, $W(s, s')$, to estimate the population cell average, $\hat{\mu}_s = \sum_{s'}W(s,s')n_{s'}^\calR \bar{Y}_{s'}$.
A multilevel model is a special case that smooths the cell averages by partially pooling together cells with the same lower-order features.
In this more general case the DRP estimator is again a weighting estimator, with adjusted weights
\[
  \tilde{\gamma}(s) = \hat{\gamma}(s) + \sum_{s'}W(s,s')(N_{s'}^\calP -n_{s'}^\calR\hat{\gamma}(s')).
\]
Here the weights are adjusted by a smoothed average of the imbalance in similar cells. In the extreme case where the weight matrix is diagonal with elements $\frac{1}{n^\calR}$, the DRP estimate reduces to the post-stratification estimate, as above. 

One important special case are tree-based methods such as those considered by \citet{Montgomery2018} and \citet{bisbee2019barp}. These methods estimate the outcome via bagged regression trees, such as random forests \citep{Breiman2001}, gradient boosted trees \citep{Friedman2001}, or Bayesian additive regression trees \citep{Chipman2010}. These approaches 
% combine the predictions of different trees, and 
can be viewed as data-adaptive weighting estimators where the weight for cell $s$ and $s'$, $W(s,s')$, is proportional to the fraction of trees where cells $s$ and $s'$ share a leaf node \citep{Athey2019_grf}.
Therefore, the DRP estimator will smooth the weights by adjusting them to correct for the imbalance in cells that share many leaf nodes.

\paragraph{Bias-variance tradeoff.}
The key difference between MRP and DRP is what role the outcome model plays, and how one chooses the model to negotiate the bias-variance tradeoff.
Because MRP-style estimators \emph{only} use the outcome model, the performance of the outcome model completely determines the performance of the estimator.
For example, in a multilevel model we want to include higher-order interaction terms in order to reduce the bias. However, this can increase the variance to an unacceptable degree, so we choose a model with lower variance and higher bias.

In contrast, with DRP the role of the model is to correct for any potential bias remaining after multilevel calibration. Because we can only \emph{approximately} post-stratify, this bias-correction is key. It also means that DRP is less reliant on the outcome model, which only needs to adequately perform bias correction.
Therefore, the bias-variance tradeoff is different for DRP, 
and we prioritize bias over variance.
By including higher order interactions or deeper trees, the model will be able to adjust for any remaining imbalance in higher order interactions after weighting.

%%%
%%% COMPARISON TO IPW
%%%
\section{Comparison to inverse propensity score weighting via multilevel modelling}
\label{sec:dual}

We now show that the multilevel calibration approach is a form of inverse propensity score weighting with a multilevel non-response model. This connection is instructive, especially for DRP, because traditional propensity score models can have steep data requirements, often requiring detailed individual-level data for both the sample and the target population \citep[see][]{chen2020dr_surveys}. By contrast, the data requirements for multilevel calibration weights are somewhat weaker, requiring aggregate data on all interactions of interest.

In particular, when we enforce \emph{exact} balance on all interactions,  multilevel calibration weights are equivalent to IPW with propensity scores estimated via a fully-saturated generalized linear model (GLM) --- and both our proposed weights and traditional IPW weights are equivalent to post-stratification weights. 
As we show, the primary difference between the multilevel calibration approach and a multilevel GLM is in how the propensity score coefficients are regularized. Through the Lagrangian dual, we will see that the multilevel calibration approach implicitly regularizes the coefficients on interactions to guarantee balance while the multilevel GLM approach does not.

\subsection{Dual relation to multilevel non-response modelling}
We begin by deriving the Lagrangian dual to the optimization problem \eqref{eq:primal}. 
By inspecting the dual, we can characterize the implicit propensity score model associated with the weights, moving smoothly between raking on margins and post-stratification. This builds on recent results noting the connection between approximate balancing weights estimators and \edit{calibrated} regularized propensity score estimation \citep[e.g.][]{Wang2019, Hirshberg2019, Zhao2019, Chattopadhyay2019,Tan2020_calibrated, benmichael2020_lor} as well as a long history linking raking weights to IPW with a propensity score that is log-linear in the first-order marginals \citep{Little1991}.

The dual problem involves optimizing over a series of Lagrange multipliers. The raking constraint induces one set of Lagrange multipliers $\beta^{(1)}$. In the same way, the approximate post-stratification objective induces an additional set of Lagrange multipliers $\beta^{(k)}$--- one for each group of higher order interactions. These dual variables are then chosen to optimize a regularized objective function.\footnote{For ease of exposition we have derived the Lagrangian dual for the the usual case where $L = 0$ and $U = \infty$. For general $L < U$,  $\gamma(s;\hat{\beta})$ will be truncated at $L$ and $U$, and the loss function will change slightly.}

\begin{proposition}
  \label{prop:dual}
  If a feasible solution to \eqref{eq:primal} exists, the Lagrangian dual problem
  with $L = 0$ and $U = \infty$
  is
  \begin{equation}
    \label{eq:dual}
    \min_{\beta} \;\; \underbrace{\frac{1}{2N}\sum_{i=1}^n\left[ R_i \max\left\{0, \sum_{k=1}^d D_{S_i}^{(k)} \cdot \beta^{(k)}\right\}^2 - \sum_{k=1}^d D_{S_i}^{(k)} \cdot \beta^k\right]}_{\text{loss function } q(\beta)} + \underbrace{\sum_{k=2}^d\frac{\lambda_k}{2}\|\beta^k\|_2^2}_{\text{regularization}},
  \end{equation}
  where $\beta = (\beta^{(1)},\ldots,\beta^{(d)})$. If $\hat{\beta}$ solves \eqref{eq:dual}, the primal weights are recovered as 
  \begin{equation}
    \label{eq:dual_to_primal}
    \hat{\gamma}(s) = \max \left\{0, \sum_{k=1}^d D_{s_i}^{(k)} \cdot \hat{\beta}^{(k)}\right\} \equiv \gamma(s; \hat{\beta}).
  \end{equation}
\end{proposition}

\noindent To connect this to propensity score estimation, we can inspect the minimizer of the expected loss, $\E[q(\beta)]$. The zero gradient condition for the expected loss is
\[
  \nabla \E[q(\beta)] = 0 \;\;\Longleftrightarrow\;\; \sum_s N_s^\calP \pi(s) \gamma(s;\beta)D_s = \sum_s N_s^\calP D_s.
\]
The unique weights that solve the expected zero gradient condition are precisely the inverse propensity weights $\gamma(s;\beta) = \frac{1}{\pi(s)}$. Therefore, the dual solution is a \emph{regularized} $M$-estimator for the propensity score, with a fully saturated propensity score model that includes all interactions.

\subsection{The role of regularization}

We now compare regularization in multilevel calibration weights versus more traditional multilevel GLM estimation for the propensity score.
These two models have the same starting point: both are $M$-estimators for the propensity score and, in the special case without regularization, both are equivalent to post-stratification weights and to each other.
Both estimators also partially pool the propensity score estimates across cells.
However, in practical settings where full post-stratification is infeasible, regularization affects the two approaches differently.
For multilevel calibration, 
the regularization in the dual problem \eqref{eq:dual} 
ensures a level of balance on interaction terms.
By contrast, for the multilevel GLM, 
the regularization 
instead controls 
a different quantity that is only indirectly relevant for estimating the population average.

To see this, we can examine the zero gradient conditions for the two approaches.
First, for multilevel calibration weights, the level of partial pooling directly relates to balance in the higher order interactions. 
The zero gradient condition for the regularized dual problem \eqref{eq:dual} implies that the imbalance in the $k$\super{th} order interactions is
\[  
  \frac{1}{N}\left\|\sum_s D_s^{(k)} n_s^\calR \gamma(s; \hat{\beta}) - \sum_s D_s^{(k)} N_s^\calP\right\|_2 = \lambda_k \|\hat{\beta}^{(k)}\|_2,
\]
Therefore $\lambda_k$ directly controls the level of balance in the $k$\super{th} order interactions, and so the level of regularization controls how far the re-weighted sample is from the target. 

We can compare this to the zero gradient condition of the propensity score $\pi(s;\hat{\beta})$ estimated via a multilevel GLM with equivalent hyper-parameters:
\[
  \frac{1}{N}\left\|\sum_{s=1}^J D_s^{(k)} n_s^\calR - \sum_{s=1}^JD_s^{(k)} \pi(s; \hat{\beta})N_s^\calP\right\|_2 = \lambda_k \|\hat{\beta}^{(k)}\|_2.
\]
Here the hyper-parameter $\lambda_k$ instead controls the difference between the observed sample counts and the expected counts under the model. This difference is only indirectly related to estimating the population means, in essence estimating the propensity score $\pi(s)$ rather than the inverse propensity score $\frac{1}{\pi(s)}$.
Therefore, while both approaches are estimators of a fully-interacted propensity score, the regularization in multilevel calibration controls an upper bound on the bias when estimating the population average $\mu$. In contrast, regularization in the multilevel GLM  provides a condition on a quantity that is incidental to estimating $\mu$.

%%%
%%% SIMULATION STUDY
%%%

\section{Simulation study calibrated to the 2016 U.S. presidential election}
\label{sec:sims}
We now evaluate the statistical behavior of the multilevel calibration and DRP estimators on simulated data based on our application to the 2016 United States Presidential election, described in Section \ref{sec:intro_application}.
We calibrate two non-response models to the response structure in this population. 
First, we fit a random forest model to predict response (i.e. inclusion in the Pew sample) with $B=500$ trees, so that the probability of responding in cell $s$ is
\[
  \pi^{\text{rf}}(s) = \sum_{s'} \frac{n_{s'}^\calR}{B}\sum_{b=1}^B \frac{\bbone\{s' \in L_b(s)\}}{|L_b(s)|}.
\] 
We also consider a fourth-order model, where the response probability for cell $s$ is
\[
  \pi^{(4)}(s) = \text{logit}^{-1}\left(\sum_{k=1}^4 \hat{\beta}^{(k)} \cdot D_s^{(k)} \right),
\] 
and the coefficient vector $\hat{\beta}$ is under-regularized so that there is poor overlap.
We similarly consider two different outcome models for presidential candidate vote choice. First, we fix the outcomes to be unchanged from the original data; second, we model the probability that unit $i$ votes Republican ($Y_i = 1$) as a fourth order logistic regression model as above, similarly under-regularized.

\begin{figure}[tb]
  \centering \includegraphics[width=.95\maxwidth]{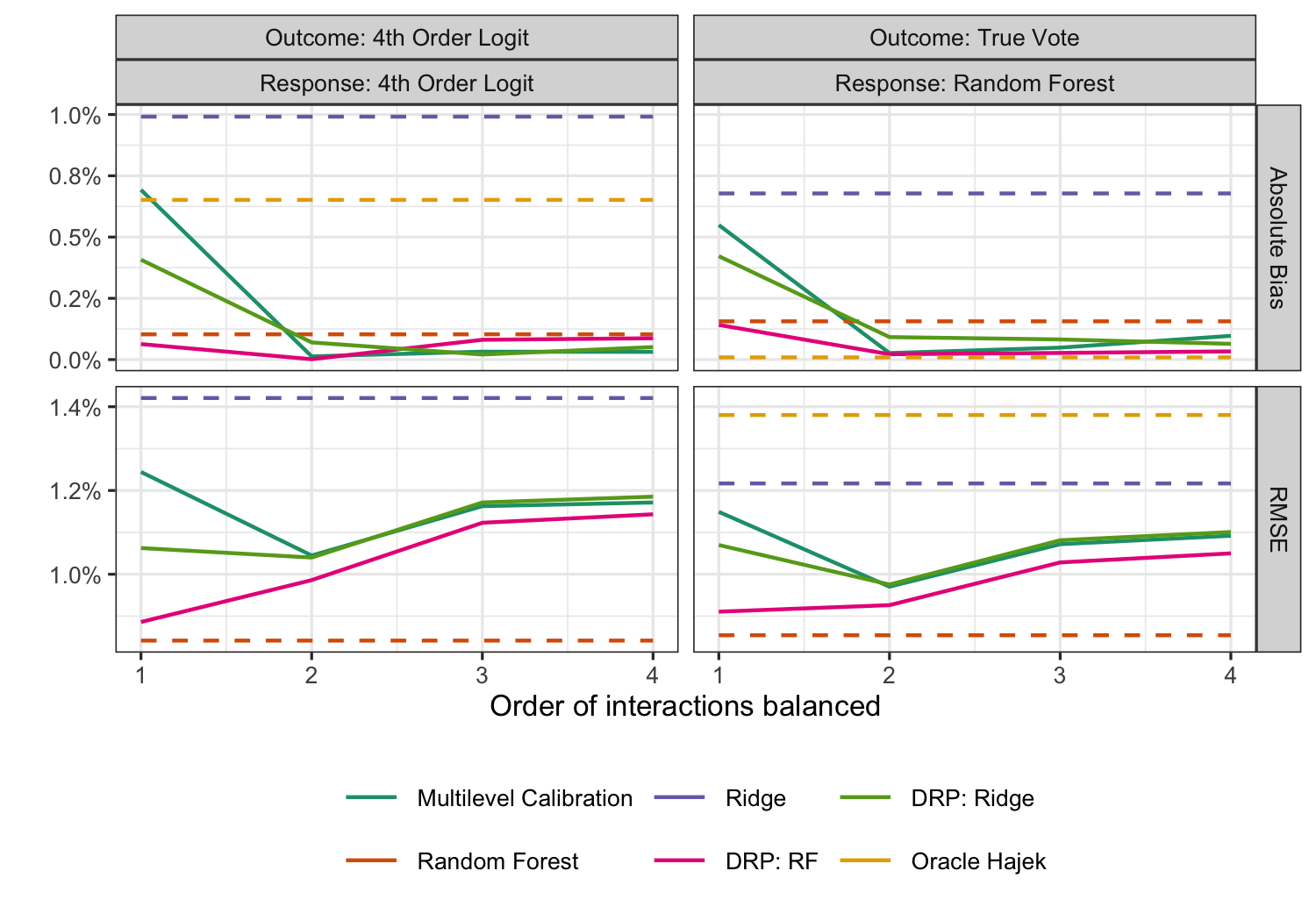}
  \caption{Bias and RMSE across 1000 simulation runs. RMSE for the oracle Horvitz-Thompson estimator in the under-regularized fourth order model (6.4\%) omitted for scale. } 
  \label{fig:bias_rmse}
\end{figure}

To generate simulation runs, we re-sample from the population with replacement, so the total number of units $N$ is fixed while the number of units within each cell $N_s^\calP$ varies. We then generate responses and outcomes according to the probabilities above using two pairs: (a) fourth order models for both the response and the outcome, and (b) a random forest response model with the true, deterministic outcomes. We consider using multilevel calibration weighting in Equation \eqref{eq:primal}, balancing first, second, third, and fourth order interactions with $\lambda^{(k)} = 1$ and setting $\lambda^{(k)} = 0$ for interactions of higher order. We also consider the DRP estimator, bias correcting with either a third-order ridge regression or a random forest, as well as MRP with these outcome models. Finally, we compare to the oracle Horvitz-Thompson estimator with the true response probabilities.

Figure \ref{fig:bias_rmse} shows the bias and root mean square error (RMSE) of these approaches across simulation runs. First looking at the bias, we see that under both data generating processes (DGPs) it is not enough to rake on margins, and there are substantial gains to balancing second and higher order interactions. Next, bias correction can provide large improvements: under both DGPs, DRP reduces the bias relative to raking on the margins alone by nearly the same degree as directly balancing higher order interaction terms. Even in the under-regularized fourth order DGP---where the oracle Horvitz-Thompson estimator performs poorly---we can significantly reduce the bias.
DRP also has reduced bias relative to MRP alone with the same outcome model.
Focusing on RMSE, we see that the decrease in bias from balancing higher order interactions outweighs the increase in variance only when balancing second order interactions, with third and fourth order interactions having a worse bias-variance trade-off. We see however, that the bias-variance trade-off for including an outcome model through DRP is favorable under both outcome models and DGPs, with the DRP estimator with a random forest outcome model and raking weights having the lowest RMSE. Finally, MRP with ridge regression has higher RMSE than multilevel calibration and DRP, while MRP with random forest (the oracle estimator for one of the DGPs) has lower RMSE.

% Case study: 
\section{2016 US Presidential election polls}
\label{sec:numerical}
We now turn to evaluating the proposed estimators in the context of 2016 US Presidential polling, as described in Section \ref{sec:intro_application}. 
We begin by showing balance gains from the multilevel calibration procedure and inspecting how bias correction through DRP affects both the point estimates and confidence intervals. We then evaluate the performance of multilevel calibration and DRP when predicting state-specific vote counts from the national pre-election survey of vote intention.

\begin{figure}[tb]
  \centering \includegraphics[width=.5\maxwidth]{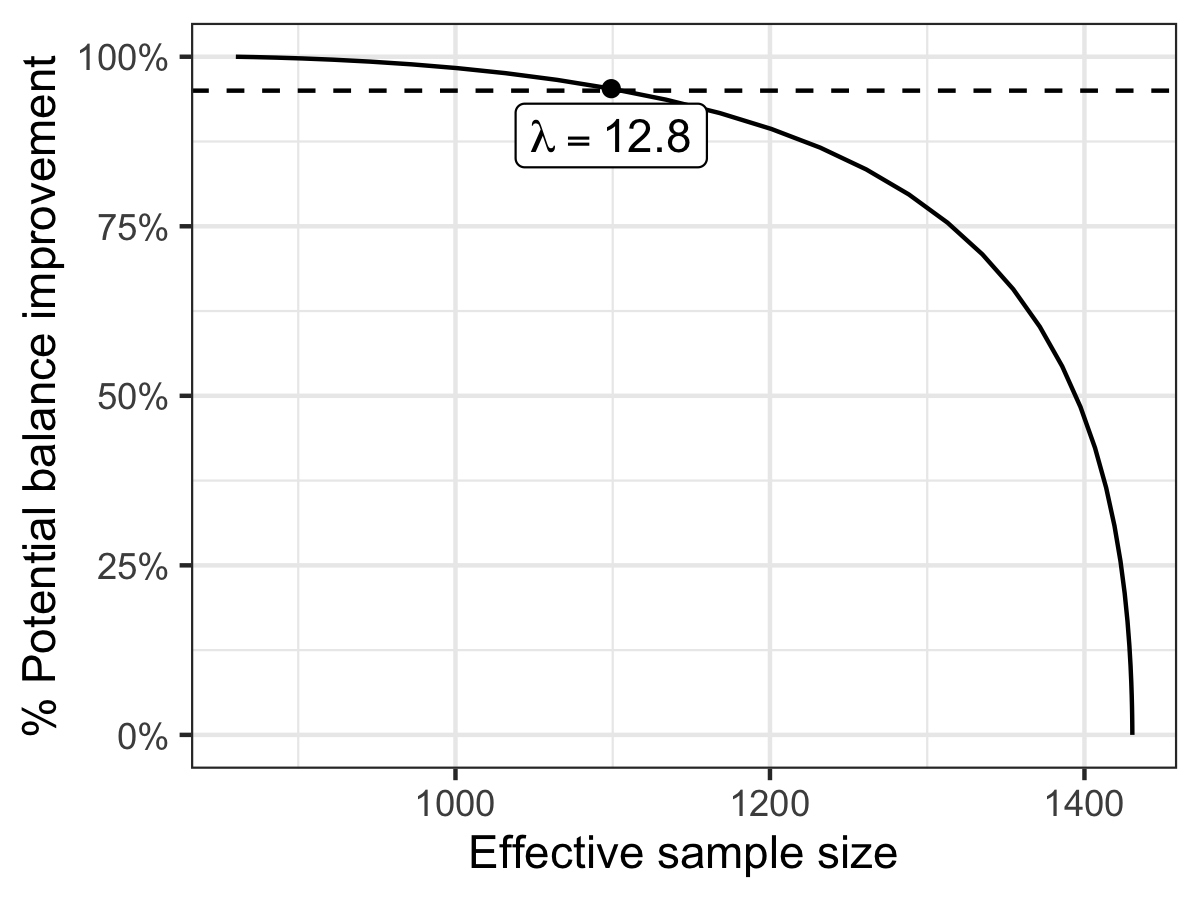}
  \caption{Difference between the re-weighted sample and the population, measured as the square root of the sum of squared imbalances for interactions $k=1,\ldots,6$, versus the effective sample size. Imbalance measures are scaled as the percent reduction in imbalance relative to raking on margins. } 
  \label{fig:balance_v_neff}
\end{figure}

We compute population cell counts $N_s^\calP$ from the post-2016 election CCES poll, limiting to those who voted in the election as indicated by a flag for a verified voter from the Secretaries of State files, and weighting according to provided CCES survey weights. We consider the balance of three different weighting estimators. First, we rake on margins for eight variables measured in both surveys, equivalent to solving \eqref{eq:primal} with $\lambda_k \to \infty$ for $k \geq 2$ and $L = 0, U = \infty$. Next, we balance up to 6\super{th} order interaction terms, setting a common hyper-parameter $\lambda_k=\lambda$ for $k=2,\ldots,6$ and $\lambda_k  \to \infty$ for $k = 7,8$. 
To select $\lambda$, we solve \eqref{eq:primal} for a series of potential values, tracing out the bias-variance trade-off in Figure \ref{fig:balance_v_neff}. We find that a value of $\lambda = 12.8$ achieves 95\% of the potential imbalance reduction  while having an effective sample size 30\% larger than the least-regularized solution.
Last, we create post-stratification weights. Due to the number of empty cells, we limit to post-stratifying on four variables, collapsed into coarser cells.\footnote{We collapse income and age to 3 levels, education to a binary indicator for greater than a high school degree, and race to a binary indicator for white.} We also consider bias-correcting the multilevel weights with DRP with (a) a fourth order ridge regression model and (b) gradient boosted trees, both tuned with cross validation.

\begin{figure}[tb]
  \centering \includegraphics[width=.95\maxwidth]{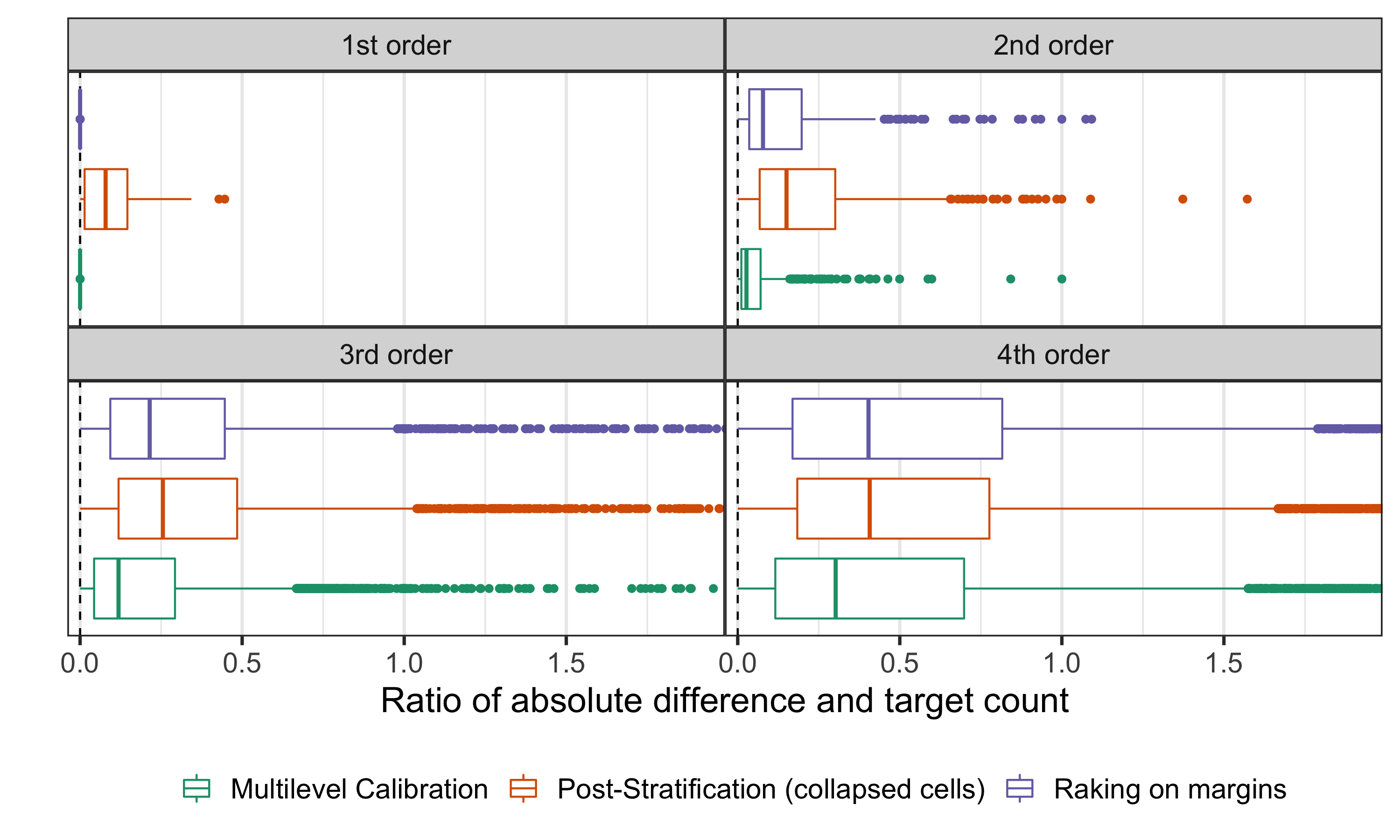}
  \caption{Covariate imbalance for interactions up to order 4, measured as the difference between the weighted and target count, divided by the target count.} 
  \label{fig:imbalance}
\end{figure}

Figure \ref{fig:imbalance} shows the imbalance when weighting by these three approaches for interactions up to order 4. To place the balance on the same scale, we divide the difference between the re-weighted sample and the population in the $j$\super{th} interaction of order $k$ by the population count, $\frac{\left|\sum_s D_{sj}^{(k)}(n_s^\calR \hat{\gamma}(s) - N^\calP)\right|}{\sum_s D_{sj}^{(k)} N_s^\calP}$. By design, both the raking and multilevel calibration weights exactly balance first order margins; however, post-stratifying on a limited set of collapsed cells does not guarantee balance on the margins of the uncollapsed cells, due to missing values. The multilevel calibration weights achieve significantly better balance on second order interactions than do the raking weights or the post-stratification weights. For higher order interactions these gains are still visible but less stark, as it becomes more difficult to achieve good balance.

This improvement in balance comes at some cost to variance. Figure \ref{fig:weight_hist} shows the empirical CDF of the respondent weights for the three approaches. The multilevel calibration weights that balance higher order interactions have a greater proportion of large weights, with a longer tail than raking or collapsed post-stratification. These large weights lead to a smaller effective sample size. The multilevel calibration weights yield an effective sample size of 1,099 for a design effect of 1.87, while raking and the collapsed post-stratification weights have effective sample sizes of 1,431 and 1,482 respectively.

\begin{figure}[tbp]
  \centering
    \begin{subfigure}[t]{0.5\textwidth}  
      \centering \includegraphics[width=.95\maxwidth]{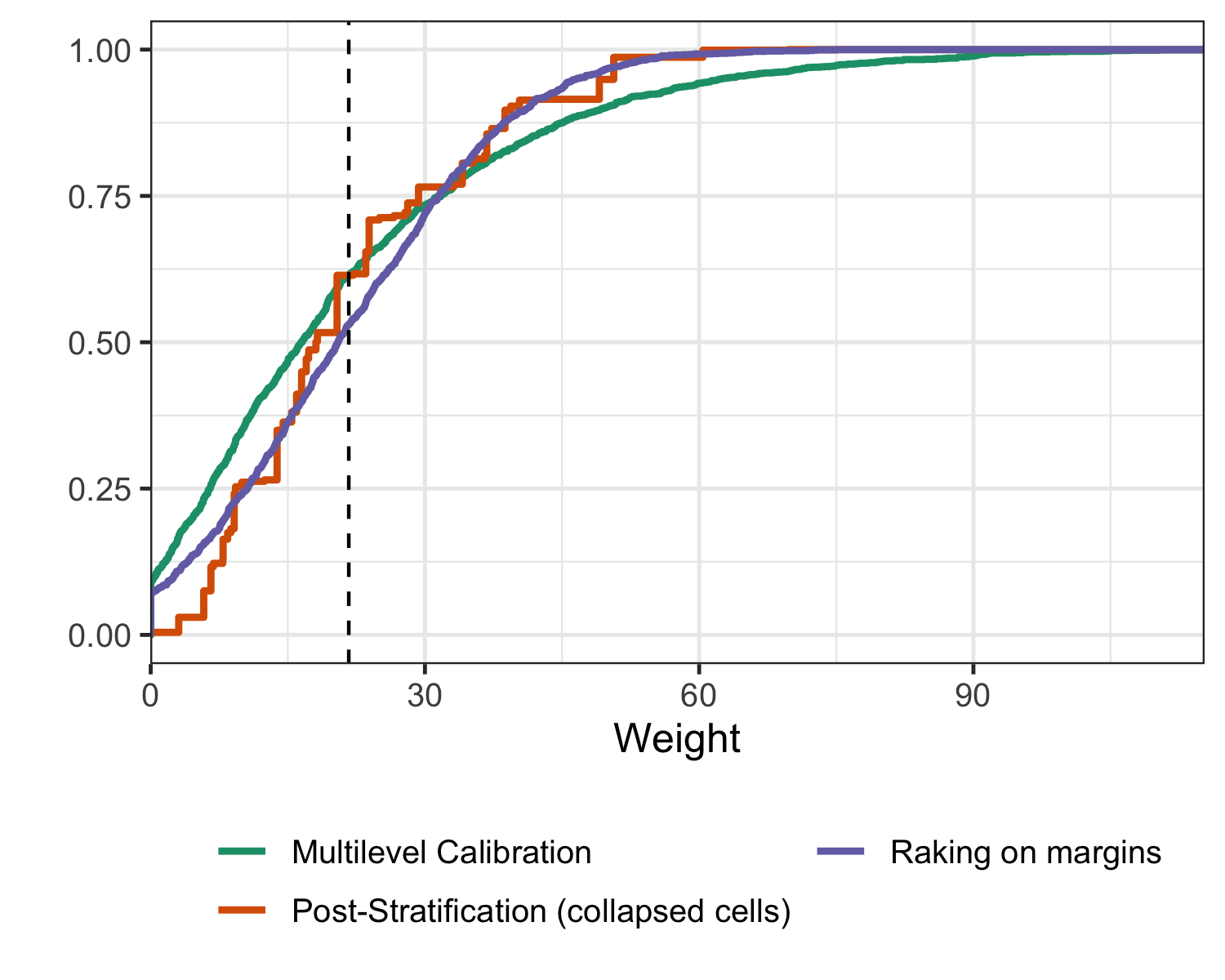}
      \caption{Empirical CDF of weights} 
      \label{fig:weight_hist}
    \end{subfigure}%
    ~
    \begin{subfigure}[t]{0.5\textwidth}  
      \centering \includegraphics[width=.95\maxwidth]{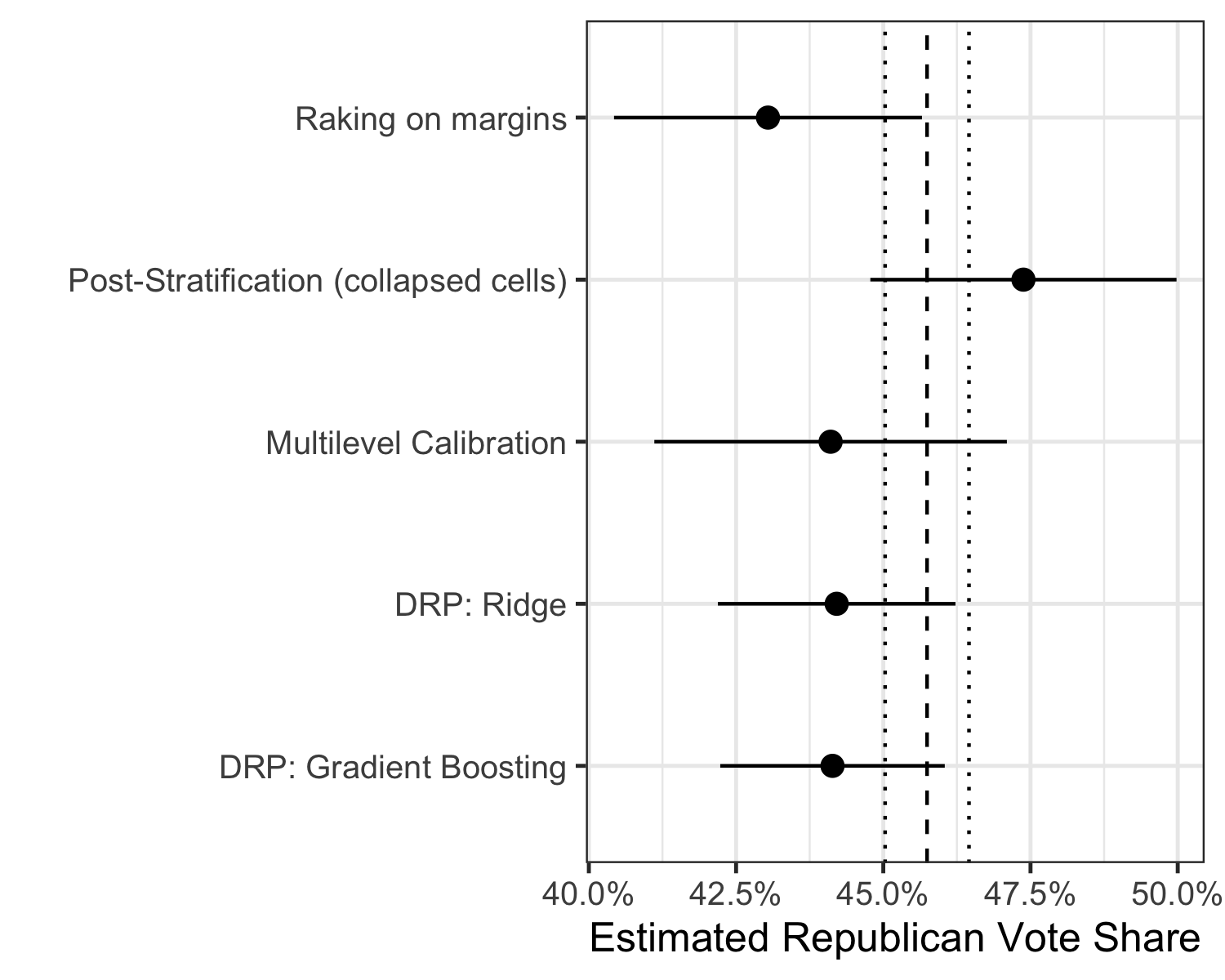}
      \caption{Predictions of Republican vote share} 
      \label{fig:res_plot}
      \end{subfigure}
\caption{(a) Distribution of weights. Dashed line indicates a uniform adjustment $\frac{N}{n}$. (b) Point estimates and approximate 95\% confidence intervals. Thick dashed line is the weighted CCES estimate, thinner dashed lines indicate the lower and upper 95\% confidence limits.}
\label{fig:res_plot_hist}
\end{figure}

Figure \ref{fig:res_plot} plots the point estimates and approximate 95\% confidence intervals for the multilevel calibration and DRP approaches, along with the estimated Republican vote share from the weighted CCES sample.
The different weights result in different predictions of the vote share, ranging from a point estimate of 42.5\% for raking to 47.5\% for post-stratification.
Additionally, the somewhat smaller effective sample size for multilevel calibration manifests itself in the standard error, leading to slightly larger confidence intervals.
The DRP estimators, bias correcting with either ridge regression or gradient boosted trees, have similar point estimates to multilevel calibration alone. This indicates that the remaining imbalances in higher order interactions after weighting in Figure \ref{fig:imbalance} do not lead to large estimated biases.
However, by including an outcome model the DRP estimators significantly reduce the standard errors.

\begin{figure}[tb]
  \centering \includegraphics[width=.95\maxwidth]{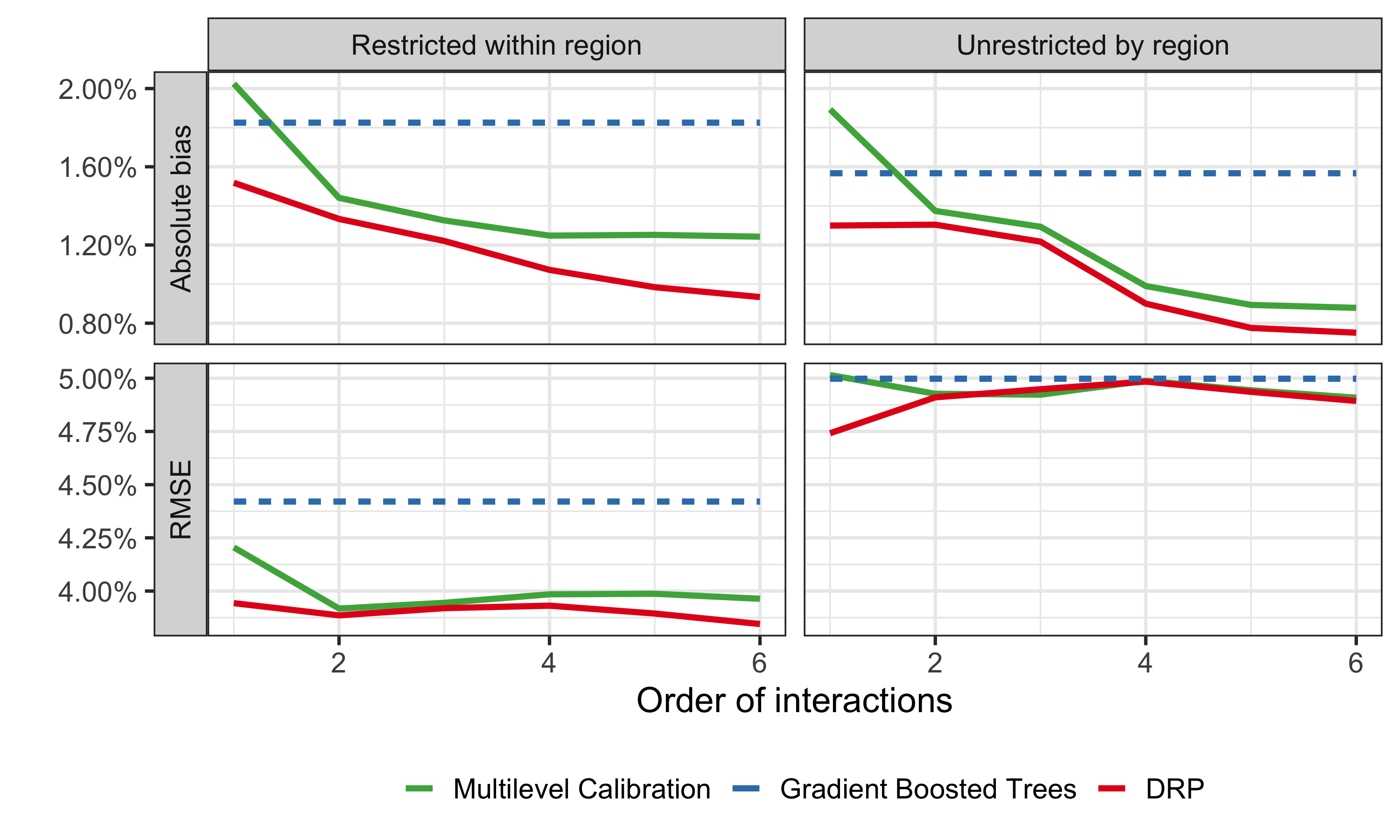}
  \caption{Absolute bias and MSE when imputing Republican vote share in 50 states from the national Pew survey, restricting to respondents in the same region and unrestricted by region.} 
  \label{fig:metrics}
\end{figure}

To empirically validate the role of balancing higher order interactions, we use the national pre-election Pew survey to predict Republican vote share within each state.
The pre-election survey was designed as a \emph{national} survey and so there are substantial differences between the sample and the state-level totals.
For each state we compute the population count vector $N^\calP$ from the weighted CCES, subset to the state of interest. Here we use a common $\lambda = 1$. We then impute the Republican vote share for that state via weighting alone and DRP with gradient boosted trees, balancing interactions up to order six. We consider both restricting the sample respondents to be in the same region as the state and including all sample respondents. 
Figure \ref{fig:metrics} shows the absolute bias and RMSE across the 50 states as the order increases from raking on first order margins to approximately balancing sixth order interactions. There are substantial gains to bias-correction through DRP when raking on the margins in terms of both bias and variance. Balancing higher order interactions also improves estimation over raking alone. And while the relative improvement of DRP over multilevel calibration diminishes as we balance higher order interactions, these gains are still apparent. Finally, while not restricting respondents by region results in lower bias across the 50 states, the higher RMSE shows that the estimates of state vote share are poor but averaging out.

%%%
%%% DISCUSSION
%%%
\section{Discussion}
\label{sec:discussion}

As recent public opinion polling has shown, differential non-response across groups defined by fine-grained higher order interactions of covariates can lead to substantial bias. While, ideally, we would address such nonresponse by post-stratifying on all interactions of important covariates simultaneously, the cost of collecting the necessary sample size is prohibitive, especially with low response rates. In practice, analysts circumvent this via ad hoc approaches, such as only adjusting for first-order marginal characteristics or collapsing cells together.

In this paper we provide two alternative approaches, \emph{multilevel calibration weighting} and \emph{Double Regression with Post-stratification (DRP)}, which provide principled ways to combine fine-grained calibration weighting and modern machine learning prediction techniques. 
The multilevel calibration weights improve on existing practice by approximately post-stratifying in a data-driven way, while at least ensuring exact raking on first order margins. DRP then takes advantage of flexible regression methods to further adjust for differences in fine-grained cells in a parsimonious way. 
For groups where the weights successfully adjust for differences in response rates, the DRP estimate is driven by the weights; 
for groups that remain over- or under-represented, DRP instead relies on a flexible regression model to estimate and adjust for remaining non-response bias.
Through theoretical, numerical, and simulation results, we find that these approaches can significantly improve estimation. Specifically, adjusting for higher-order interactions with multilevel calibration has much lower bias than ignoring them by only raking on the first-order margins. Incorporating flexible outcome estimators such as multilevel regression or tree-based approaches in our DRP estimator further improves upon weighting alone.

However, our proposal is certainly not a panacea, and important questions remain.
First, while we choose the value of the hyper-parameters by tracing out the bias-variance trade-off, it might be preferable to select them via data-adaptive measures. For example, \citet{Wang2019} propose a cross-validation style approach that takes advantage of the Lagrangian dual formulation. It may be possible to use such approaches in this setting.

Second, the key assumption is that outcomes are missing at random within cells. While we never expect this to be entirely true, it allows us to make progress on estimation, and with granular enough groups, we may hope that this assumption is approximately true. It is important then to characterize how our conclusions would change if this assumption is violated, and the response and the outcome are correlated even within cells. This form of \emph{sensitivity analysis} is common for matching estimators in observational studies, and has been proposed for inverse probability weighting \citep{Zhao2019_sens} and balancing estimators \citep{soriano_sens}. We leave adapting these approaches to this setting as future work. 

Third, with many higher order interactions it is difficult to find good information on population targets. We may have to combine various data sources collected in different manners, or impute unknown cells in the target population, and uncertainty in the population targets can also lead to increased variance \citep[see][for a recent review]{caughey_2020_elements}. 
Fourth, during the survey process we can obtain very detailed auxiliary information on survey respondents that we cannot obtain for the population, even marginally. Incorporating this sort of auxiliary information into the estimation procedure will be important to future work. 

Fifth, the asymptotic theory requires strong assumptions on the total number of cells.
It may be possible to weaken these assumptions by incorporating the outcome model into the analysis of the multilevel calibration weighting approach, using tools from Reproducing Kernel Hilbert Space theory \citep[see, e.g.][]{Hirshberg2019}.
Additionally, with categorical covariates, the non-response probability is completely determined by all of the interactions; continuous covariates will require stronger assumptions.

Finally, we propose principled procedures to account for non-response bias due to differences in response rates in groups defined by higher order interactions. While this is a pernicious problem, especially with lowered response rates, it is far from the only form of non-response bias, let alone the only difficulty with modern surveys. We therefore view multilevel calibration and DRP as only one part of the analyst's toolkit, supplementing design and data considerations.

%%%
%%% BIBLIOGRAPHY
%%%
\clearpage
\singlespacing
\bibliographystyle{chicago}
\bibliography{citations}

%%% APPENDIX
%%%

\clearpage
\appendix

\renewcommand\thefigure{\thesection.\arabic{figure}}    
\setcounter{figure}{0}  
\renewcommand\theassumption{A.\arabic{assumption}}    
\setcounter{assumption}{0}

\clearpage
\section{Proofs and derivations}
\label{sec:proofs}

\begin{assumption}
\label{a:regularity}
    There is a sequence of populations of size $N$ with $N \to \infty$ such that
  \begin{enumerate}[label = (\alph*),ref={\theassumption\alph*}]
    \item \label{a:cell_rate} The condition number of $D$, $\kappa \equiv \|D^{-1}\|_2\|D\|_2$, and the number of cells $J$ satisfy $\frac{\kappa^2J}{(\pi^\ast N)^\alpha} \to c$ for some constant $c$ and for an $0 \leq \alpha < 1$
    \item The response variables $R_i$ are independent.
    \item $\pi(s) \geq \pi^\ast > 0$ for all $N$, where $\frac{1}{{\pi^\ast}^2 N} \to 0$ as $N \to \infty$.
    \item \label{a:finite_second_moment} The residuals $\varepsilon_i \equiv Y_i - \mu_{S_i}$ satisfy $\frac{1}{N}\sum_{i=1}^N\varepsilon_i^2 < \infty$ for all population sizes $N$.
    \item \label{a:cell_var} The maximum variance across cells conditional on the cell counts, $\sigma^2 \equiv \max_s \sigma_s^2 = \max_s \Var\left(\bar{\varepsilon}_s \mid n_s^\calR \right)$ is $o_p\left(\left(\pi^\ast N\right)^{-\frac{\alpha}{2}}\right)$.
    \item \label{a:not_super_eff} For a random variable $Z = o_p\left(\frac{1}{{\pi^\ast}\sqrt{N}}\right)$, the variance of the oracle estimator $V = \frac{1}{N^2}\sum_i\frac{\pi_i (1-\pi_i)}{\pi(S_i)}\varepsilon_i^2$ satisfies $\frac{Z}{\sqrt{V}} = o_p(1)$.
    \item We find $\hat{\gamma}$ via the modified problem
    \begin{equation}
    \label{eq:primal_unreg}
    \begin{aligned}
    \min_{\gamma \in \R^J} \;\; & \sum_{k=1}^d \left\|\sum_s D_s^{(k)} n_s^\calR \gamma(s) - D_s^{(k)} N_s^\calP \right\|_2^2\\
  \text{subject to} \;\; & 0 \leq \gamma(s) \leq 1\;\;\; \forall s=1,\ldots J.
  \end{aligned}
\end{equation}
  \end{enumerate}
\end{assumption}

\begin{lemma}
\label{lem:balance_bound}
    Let $\kappa \equiv \|D^{-1}\|_2\|D\|_2$ be the ratio of the maximum and minimum singular values of $D$. The solution to \eqref{eq:primal_unreg} satisfies
    
    \[
    \sqrt{\sum_{s}\left(n_s^\calR \hat{\gamma}(s) - N_s^\calP\right)^2} \leq \kappa \sqrt{\sum_s \left(\frac{n_s^\calR}{\pi(s)} - N_s^\calP\right)^2}
    \]
\end{lemma}

\begin{proof}[Proof of Lemma \ref{lem:balance_bound}]
Slightly abusing notation, denote $\frac{1}{\pi} \in (0,1)^J$ as the vector of inverse response probabilities for each cell. $\frac{1}{\pi}$ is feasible for optimization problem \eqref{eq:primal_unreg}, and so 
\[
\begin{aligned}
  \frac{1}{\|D^{-1}\|_2} \|\diag(n^\calR)\hat{\gamma} - N^\calP\|_2 &\leq \|D'(\diag(n^\calR)\hat{\gamma} - N^\calP)\|_2\\
  & \leq \left\|D'\left(\diag(n^\calR) \frac{1}{\pi} - N^\calP\right)\right\|_2\\
  &\leq \|D\|_2\left\|\diag(n^\calR) \frac{1}{\pi} - N^\calP\right\|_2
\end{aligned}
\]
Multiplying by $\|D^{-1}\|_2$ gives the result.
\end{proof}

\begin{lemma}
\label{lem:pscore_balance}
Let $\pi^\ast = \min_s \pi(s)$. For any $\delta > 0$,

\[
\frac{1}{N}\sqrt{\sum_s \left(N_s^\calP - \frac{n_s^\calR}{\pi(s)}\right)^2} \leq \frac{1}{\pi^\ast \sqrt{N}}\left(\sqrt{J \log 5} + \delta\right),
\]
with probability at least $1 - \exp\left(-2{\pi^\ast}^2 N \delta^2\right)$.

\end{lemma}

\begin{proof}[Proof of Lemma \ref{lem:pscore_balance}]
Since $R_i \in \{0,1\}$ is bounded, it is sub-Guassian with scale parameter $\frac{1}{2}$. and because they are independent, $\frac{N_s^\calP}{N} - \frac{n_s^\calR}{N\pi(s)} = N_s^\calP - \frac{1}{\pi(s)}\sum_{S_i = s} R_i$ is a mean-zero sub-Gaussian random variable with scale parameter $\frac{\sqrt{N_s^\calP}}{2\pi(s)N} \leq \frac{1}{2\pi^\ast \sqrt{N}}$. Now by a discretization argument \citep[\S~9.6]{wainwright_2019}, we have that
\[
\frac{1}{N}\sqrt{\sum_s \left(N_s^\calP - \frac{n_s^\calR}{\pi(s)}\right)^2} \geq \frac{1}{\pi^\ast \sqrt{N}}\left(\sqrt{J \log 5} + \delta\right)
\]
with probability at most $\exp\left(-2{\pi^\ast}^2N\delta^2\right)$. This completes the proof.
\end{proof}

\begin{lemma}
    \label{lem:bias_rate}
    If $\sqrt{\sum_s \left(\hat{\mu}_s - \mu_s\right)^2} = o_p\left(\left(\pi^\ast N\right)^{-\frac{\alpha}{2}}\right)$, then 
    \[
    \frac{1}{N}\sum_s \left(\hat{\mu}_s - \mu_s \right)\left(n_s^\calR \hat{\gamma}(s) - N_s^\calP\right) = o_p\left(\frac{1}{\pi^\ast \sqrt{N}}\right)
    \]
\end{lemma}

\begin{proof}[Proof of Lemma \ref{lem:bias_rate}]
First, note that by Cauchy-Schwartz,
\[
\frac{1}{N}\sum_s \left(\hat{\mu}_s - \mu_s \right)\left(n_s^\calR \hat{\gamma}(s) - n_s^\calP\right) \leq \sqrt{\sum_s \left(\hat{\mu}_s - \mu_s\right)^2} \frac{1}{N}\sqrt{\sum_s \left(n_s^\calR\hat{\gamma}(s) - N_s^\calP\right)^2}
\]
From Lemma \ref{lem:pscore_balance}, the term on the right is $O_p\left(\frac{\kappa}{\pi}\sqrt{\frac{J}{N}}\right)$. By Assumption \ref{a:cell_rate}, this is $O_p\left(\frac{1}{\pi^{1 - \alpha/2} N ^{1/2 - \alpha / 2}}\right)$. Now since $\sqrt{\sum_s \left(\hat{\mu}_s - \mu_s\right)^2} = o_p\left(\left(\pi^\ast N\right)^{-\frac{\alpha}{2}}\right)$, the product is $o_p\left(\frac{1}{\pi^\ast \sqrt{N}}\right)$
\end{proof}

\begin{lemma}
    \label{lem:noise_converges}
    Under Assumption \ref{a:cell_var}, the solution to \eqref{eq:primal_unreg}, $\hat{\gamma}$ satisfies
    \[
    \frac{1}{N}\sum_s \hat{\gamma}(s) n_s^\calR \bar{\varepsilon}_s = \frac{1}{n}\sum_{i=1}^N \frac{R_i}{\pi(S_i)} \varepsilon_i + o_p\left(\frac{1}{\pi^\ast\sqrt{N}}\right)
    \]
\end{lemma}

\begin{proof}[Proof of Lemma \ref{lem:noise_converges}]
    First, we write the noise term as
    \[
        \frac{1}{N}\sum_s \hat{\gamma}(s) n_s^\calR \bar{\varepsilon}_s = \frac{1}{n}\sum_{i=1}^N \frac{R_i}{\pi(S_i)} \varepsilon_i + \frac{1}{n}\sum_{i=1}^N R_i \left(\hat{\gamma}(s) - \frac{1}{\pi(S_i)}\right) \varepsilon_i.
    \]
The variance of the second term, conditional on the cell counts $n^\calR$ is
\[
\begin{aligned}
  \Var\left(\frac{1}{n}\sum_{i=1}^N R_i \left(\hat{\gamma}(s) - \frac{1}{\pi(S_i)}\right) \varepsilon_i \mid n^\calR \right) & = \frac{1}{N^2}\sum_s \left(\hat{\gamma}(s) - \frac{1}{\pi(s)}\right)^2 {n_s^\calR}^2\sigma_s^2\\
  &\leq \frac{\sigma^2}{N^2}\sum_s \left(\hat{\gamma}(s) - \frac{1}{\pi(s)}\right)^2 {n_s^\calR}^2  
\end{aligned}
\]
So by Chebyshev's inequality, conditional on the cell counts $n^\calR$ we have that with probability at least $1 - \delta$,
\[
\left|\frac{1}{n}\sum_{i=1}^N R_i \left(\hat{\gamma}(s) - \frac{1}{\pi(S_i)}\right) \varepsilon_i \right| \leq \frac{\sigma}{N\sqrt{\delta}}\sqrt{\sum_s\left(\hat{\gamma}(s) - \frac{1}{\pi(s)}\right)^2 {n_s^\calR}^2}.
\]

Now notice that 

\[
\begin{aligned}
  \sqrt{\sum_s\left(\hat{\gamma}(s) - \frac{1}{\pi(s)}\right)^2 {n_s^\calR}^2} & = \left\|\diag(n^\calR)\left(\hat{\gamma} - \frac{1}{\pi}\right)\right\|_2\\
  & = \left\|\diag(n^\calR)\hat{\gamma} -N^\calP + N^\calP- \diag(n^\calR)\frac{1}{\pi}\right\|_2\\
  & \leq \left\|\diag(n^\calR)\hat{\gamma} -N^\calP\right\|_2 + \left\|N^\calP- \diag(n^\calR)\frac{1}{\pi}\right\|_2
\end{aligned}
\]

From Lemma \ref{lem:balance_bound} we can further bound this by

\[
\sqrt{\sum_s\left(\hat{\gamma}(s) - \frac{1}{\pi(s)}\right)^2 {n_s^\calR}^2} \leq (1 + \kappa)\left\|N^\calP- \diag(n^\calR)\frac{1}{\pi}\right\|_2
\]
Following the Proof of Lemma \ref{lem:bias_rate}, by Lemma \ref{lem:pscore_balance} and Assumption \ref{a:cell_rate} this is $O_p\left(\frac{1}{\pi^{1 - \alpha/2} N ^{1/2 - \alpha / 2}}\right)$. Noting that by Assumption \ref{a:cell_var} $\sigma = o_p\left(\left(\pi^\ast N\right)^{-\frac{\alpha}{2}}\right)$, shows that this remainder term is $o_p\left(\frac{1}{\pi^\ast \sqrt{N}}\right)$.

\end{proof}

\begin{proof}[Proof of Theorem \ref{thm:asymp_normal}]
 First, we write $\hat{\mu}^\drp(\hat{\gamma}) - \mu$ as
 \[
 \hat{\mu}^\drp(\hat{\gamma}) - \mu = \frac{1}{N}\sum_s \left(\hat{\mu}_s - \mu_s\right)\left(n_s^\calR \hat{\gamma}(s) - N_s^\calP\right) + \frac{1}{N}\sum_{i=1}^N R_i\hat{\gamma}(S_i)\varepsilon_i
 \]
 From Lemma \ref{lem:bias_rate}, the first term is $o_p\left(\frac{1}{\pi^\ast\sqrt{N}}\right)$ and from Lemma \ref{lem:noise_converges} the second term is $\frac{1}{N}\sum_i \frac{R_i}{\pi(S_i)}\varepsilon_i + o_p\left(\frac{1}{\pi^\ast \sqrt{N}}\right)$. Combining these gives the first result. Assumption \ref{a:not_super_eff} combined with an application of Slutsky's theorem gives the second result.
\end{proof}

\begin{proof}[Proof of Proposition \ref{prop:dual}]
We begin by re-writing the optimization problem \eqref{eq:primal} with $L = 0$ and $U = \infty$ in terms of auxiliary covariates $\calE^{(k)} \equiv \sum_s D_s^{(k)} n_s^\calR \gamma(s) - D_s^{(k)}N_s^\calP$. The optimization problem becomes

\[
  \begin{aligned}
    \min_{\gamma \in \R^J} \;\; & \sum_{k=2}^d \frac{1}{2\lambda_k} \left\|\calE^{(k)} \right\|_2^2 + \frac{1}{2}\sum_{s} n_s^\calR \gamma(s)^2\\
  \text{subject to} \;\; & \sum_s D_s^{(1)} n_s^\calR \gamma(s) = \sum_s D_s^{(1)} N_s^\calP\\
  & \sum_s D_s^{(k)} n_s^\calR \gamma(s) - D_s^{(k)}N_s^\calP -\calE^{(k)} = 0\\
  & 0 \leq \gamma(s) \;\;\; \forall s=1,\ldots J.
  \end{aligned}
\]
The Lagrangian is
\[
\calL(\gamma, \calE, \beta) \equiv \sum_{s=1}^J \frac{1}{2} n_s^\calR \gamma(s)^2 - n_s^\calR \gamma(s) \sum_{k=1}^d D_s^{(k)} \cdot \beta^{(k)} + N_s^\calP \sum_{k=1}^d D_s^{(k)} \cdot \beta^{(k)} + \sum_{k=2}^d \frac{1}{2\lambda_k}\|\calE^{(k)}\|_2^2 - \calE^{(k)} \cdot \beta^{(k)} 
\]
The dual problem maximizes the Lagrangian over the domain of $\gamma$ and $\calE$, so
\[
\begin{aligned}
  q(\beta) & = -\min_{0 \leq \gamma(s), \calE} \calL(\gamma, \calE, \beta)\\
  & = \sum_{s=1}^J n_s\calR \min_{0 \leq \gamma(s)} \left\{\frac{1}{2}\gamma(s)^2 - \gamma(s) \sum_{k=1}^d D_s^{(k)} \cdot \beta^{(k)}\right\} + N_s^\calP \sum_{k=1}^d D_s^{(k)} \cdot \beta^{(k)} + \sum_{k=2}^d \min_{\calE^{(k)}}\left\{\frac{1}{2\lambda_k}\|\calE^{(k)}\|_2^2 - \calE^{(k)} \cdot \beta^{(k)} \right\}\\
  & = \frac{1}{2}\sum_{s=1}^J n_s^\calR \max\left\{0, \sum_{k=1}^d D_{S_i}^{(k)} \cdot \beta^{(k)}\right\}^2 - N_s^\calP \sum_{k=1}^d D_s^{(k)} \cdot \beta^{(k)} + \sum_{k=2}^d \frac{\lambda_k}{2} \|\beta^{(k)}\|_2^2
\end{aligned}
\]
Since there exists a feasible solution to \eqref{eq:primal} by assumption, by Slater's condition $\min_\beta q(\beta)$ is equivalent to the solution to the primal problem. The solution to the inner minimization shows that the primal and dual variables are related by $\hat{\gamma}(s) = \max\left\{0, \sum_{k=1}^d D_s^{(k)} \cdot \hat{\beta}^{(k)}\right\}$.

\end{proof}

%%%

\end{document}